\documentclass[11pt,a4paper]{article}

\usepackage{a4wide}
\usepackage{authblk}
\usepackage{amsmath, amsthm, amsfonts, mathtools}
\usepackage[algoruled, vlined, linesnumbered]{algorithm2e}
\usepackage{tikz}
\usepackage{graphics,graphicx}
\usepackage[numbers]{natbib}

\usepackage{times}
\usepackage{amssymb}
\usepackage{enumerate}
\usepackage{bm}

\usepackage{tikz}

\usepackage{amsfonts}

\usepackage{amsthm, amsmath,amsfonts, amssymb}

\newtheorem{theorem}{Theorem}
\newtheorem{corollary}{Corollary}

\newtheorem{proposition}{Proposition}
\newtheorem{definition}{Definition}

\begin{document}

\title{The Price of Uncertainty in Present-Biased Planning\thanks{Work supported by the European Research Council, Grant Agreement No. 691672.}}
\author[1]{Susanne Albers}
\author[1]{Dennis Kraft}
\affil[1]{Department of Computer Science, Technical University of Munich}
\affil[ ]{{\em \{albers, kraftd\}@in.tum.de}}
\date{}
\maketitle


\begin{abstract}
The tendency to overestimate immediate utility is a common cognitive bias.
As a result people behave inconsistently over time and fail to reach long-term goals.
Behavioral economics tries to help affected individuals by implementing external incentives.
However, designing robust incentives is often difficult due to imperfect knowledge of the parameter $\beta \in (0,1]$ quantifying a person's present bias.
Using the graphical model of Kleinberg and Oren~\cite{KO}, we approach this problem from an algorithmic perspective.
Based on the assumption that the only information about $\beta$ is its membership in some set $B \subset (0,1]$, we distinguish between two models of uncertainty: one in which $\beta$ is fixed and one in which it varies over time.
As our main result we show that the conceptual loss of efficiency incurred by incentives in the form of penalty fees is at most $2$ in the former and $1 + \max B/\min B$ in the latter model.
We also give asymptotically matching lower bounds and approximation algorithms.
\end{abstract}

\section{Introduction}
\label{sec:intro}

Many goals in life such as losing weight, passing an exam or paying off a loan require long-term planning.
But while some people stick to their plans, others lack self-control; they eat unhealthy food, delay their studies and take out new loans.
In behavioral economics the tendency to change a plan for no apparent reason is known as {\em time-inconsistent behavior}.
The questions are, what causes these inconsistencies and why do they affect some more than others?
A common explanation is that people make present biased decisions, i.e., they assign disproportionately greater value to the present than to the future.
In~this simplifying model a person's behavior is the mere result of her present bias and the setting in which she is placed.
However, the interplay between these two factors is intricate and sometimes counter-intuitive as the following example demonstrates:

Consider two runners Alice and Bob who have two weeks to prepare for an important race.
Each week they must choose between two types of workout.
Type $A$ always incurs an effort of $1$, whereas type $B$ incurs an effort of $3$ in the first and $9$ in the second week.
Since $A$ offers less preparation than $B$, Alice and Bob's effort in the final race is $13$ if they consistently choose $A$ and $1$ if they consistently choose $B$.
Furthermore, $A$ and $B$ are incompatible in the sense that switching between the two will result in an effort of $16$ in the final race.
Figure~\ref{fig:ex} models this setting as a directed acyclic graph $G$ with terminal nodes $s$ and $t$.
The intermediate nodes $v_X$ and $v_{XY}$ represent a person's state after completing the workouts $X,Y \in \{A,B\}$.
To move forward with the training, Alice and Bob must perform the tasks associated with the edges of $G$, i.e., complete workouts and run the race.
Looking at $G$ it becomes clear that two consecutive workouts of type $B$ are the most efficient routine in the long run.
However, this is not necessarily the routine a present biased person will choose.

For instance, assume that Alice and Bob discount future costs by a factor of $a = 1/2 - \varepsilon$ and $b = 1/2 + \varepsilon$ respectively.
We call $a$ and $b$ their present bias.
At the beginning of the first week Alice and Bob compare different workout routines.
From Alice's perspective two workouts of type $A$ are strictly more preferable to two workouts of type $B$ as she anticipates an effort of $1 + a(1 + 13) = 8 - 14\varepsilon$ for the former and $3 + a(9 + 1) = 8 - 10\varepsilon$ for the latter.
A similar calculation for Bob shows that he prefers two workouts of type $B$.
Considering that neither Alice nor Bob finds a mix of $A$ and $B$ particularly interesting at this point, we conclude that Alice chooses $A$ in the first week and Bob $B$.
However, come next week, Bob expects an effort of $1 + b16 = 8 + 16\varepsilon$ for $A$ and $9 + b = 19/2 + \varepsilon$ for $B$. 
Assuming $\varepsilon$ is small enough, $A$ suddenly becomes Bob's preferred option and he switches routines.
Alice on the other hand has no reason to change her mind and sticks to $A$.
As a result she pays much less than Bob during practice and in the final race.
This is remarkable considering that her present bias is only marginally different from Bob's.
Moreover, it seems surprising that only Bob behaves inconsistently, although he is less biased than Alice.

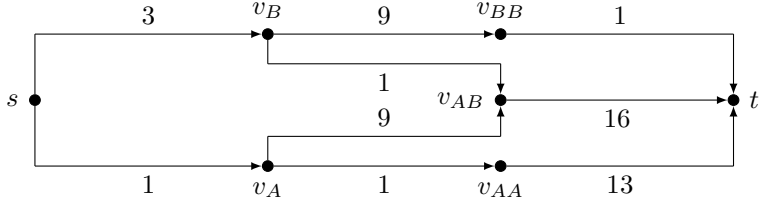
\begin{figure}[t]
	\center
	\begin{tikzpicture}[scale=0.875, nst/.style={draw,circle,fill=black,minimum size=4pt,inner sep=0pt]}, est/.style={draw,>=latex,->}]
		
		\node[nst] (v1) at (0,0) [label=left:\small{$s$}] {};
		\node[nst] (v2) at (3.5,-1) [label=below:\small{$v_A$}] {};
		\node[nst] (v3) at (3.5,1) [label=above:\small{$v_B$}] {};
		\node[nst] (v4) at (7,-1) [label=below:\small{$v_{AA}$}] {};
		\node[nst] (v5) at (7,0) [label=left:\small{$v_{AB}$}] {};
		\node[nst] (v6) at (7,1) [label=above:\small{$v_{BB}$}] {};
		\node[nst] (v7) at (10.5,0) [label=right:\small{$t$}] {};		
			
		\path (v1) edge (0,-1)
		(v1) edge (0,1)
		(0,-1) edge[est] node [below] {\small{$1$}} (v2)
		(0,1) edge[est] node [above] {\small{$3$}} (v3)
		(v2) edge (3.5,-0.55)
		(v3) edge (3.5,0.55)
		(v2) edge[est] node [below] {\small{$1$}} (v4)
		(v3) edge[est] node [above] {\small{$9$}} (v6)
		(3.5,-0.55) edge node [above] {\small{$9$}} (7,-0.55)
		(3.5,0.55) edge node [below] {\small{$1$}} (7,0.55)
		(7,-0.55) edge[est] (v5)
		(7,0.55) edge[est] (v5)
		(v4) edge node [below] {\small{$13$}} (10.5,-1)
		(v5) edge[est] node [below] {\small{$16$}} (v7)
		(v6) edge node [above] {\small{$1$}} (10.5,1)
		(10.5,-1) edge[est] (v7)
		(10.5,1) edge[est] (v7);
		\end{tikzpicture}
	\caption{Task graph of the running scenario}\label{fig:ex}
\end{figure}

\subsection{Related Work}
Traditional economics and game theory are based on the assumption that people maximize their utility in a rational way.
But despite their prevalence, these assumptions disregard psychological aspects of human decision making observed in empirical and experimental research~\cite{Laa}.
For instance, time-inconsistent behavior such as procrastination seems paradox in the light of traditional economics.
Nevertheless, it can be explained readily by a tendency to overestimate immediate utility in long-term planning, see e.g.~\cite{OR2}.
By studying such cognitive biases, behavioral economics tries to obtain more realistic economic~models.

A significant amount of research in this field has been devoted to {\em temporal discounting} in general and {\em quasi-hyperbolic discounting} in particular, see~\cite{FLO} for a survey.
The quasi-hyperbolic discounting model proposed by Laibson~\cite{Lai} is characterized by two parameters: the {\em present bias} $\beta \in (0,1]$ and the {\em exponential discount rate} $\delta \in (0,1]$.
People who plan according to this model have an accurate perception of the present, but scale down any costs and rewards realized $t \geq 1$ time units in the future by a factor of $\beta \delta^t$.
To keep our work clearly delineated in scope, we adopt Akerlof's model of quasi-hyperbolic discounting~\cite{A} and make the following two assumptions:
First, we focus on the present bias $\beta$ and set the exponential discount rate to $\delta = 1$.
Secondly, we assume people to be {\em naive} in the sense that they are unaware of their present bias and only optimize their current perceived utility when making a decision.
Note that Alice and Bob from the previous example behave like agents in Akerlof's model for a present bias of $\beta = 1/2 - \varepsilon$ and $\beta = 1/2 + \varepsilon$ respectively.

Until recently the economic literature lacked a unifying and expressive framework for analyzing time-inconsistent behavior in complex social and economic settings.
Kleinberg and Oren closed this gap by modeling the behavior of naively present biased individuals as a planning problem in task graphs like the one depicted in Figure~\ref{fig:ex}~\cite{KO}.
We introduce this framework formally in Section~\ref{sec:model}.
As a result of Kleinberg and Oren's work, an active line of research at the intersection of computer science and behavioral economics has emerged.
For instance, the graphical model has been used to systematically analyze different types of quasi-hyperbolic discounting agents such as {\em sophisticated} agents who are fully or partially aware of their present bias~\cite{KOR1} and agents whose present bias varies randomly over time~\cite{GILP}.
Furthermore, the graphical model was used to shed light on the interplay between temporal biases and other types of cognitive biases~\cite{KOR2}.

The graphical model is of particular interest to us as it provides a natural framework for a design problem frequently encountered in behavioral economics.
Given a certain social or economic setting, the problem is to improve a time-inconsistent person's performance via various sorts of {\em incentives}, such as  monetary rewards, deadlines or penalty fees, see e.g.~\cite{OR3}.
Using the graphical model, Kleinberg and Oren demonstrate how a strategic choice reduction can incentivize people to reach predefined goals~\cite{KO}.
To implement their incentives, they simply remove the corresponding edges from the task graph.
However, there is a computational drawback to this approach.
As we have shown in previous work, an optimal set of edges to remove from a task graph with $n$ nodes is NP-hard to approximate within a factor less than $\sqrt{n}/3$~\cite{AK}.
A more general form of incentives avoiding these harsh complexity theoretic limitations are penalty fees. 
In the graphical model penalty fees are at least as powerful as choice reduction and admit a polynomial time $2$-approximation~\cite{AK2}.

\subsection{Incentive Design for an Uncertain Present Bias}
Frederick, Loewenstein and O'Donoghue have surveyed several attempts to estimate people's temporal discount functions~\cite{FLO}.
But as estimates differ widely across studies and individuals, the difficulty of predicting a person's temporal discount function becomes apparent.
Clearly, this poses a serious challenge for the design of reliable incentives.
After all, Alice and Bob's scenario demonstrates how arbitrarily small changes in the present bias can cause significant changes in a person's behavior.
In this work we address the effects of incomplete information about a person's present bias in two different notions of uncertainty.

In Section~\ref{sec:upb} we consider naive individuals whose exponential discount rate is $\delta = 1$, but whose present bias $\beta$ is unknown.
The only prior information we have about $\beta$ is its membership in some larger set $B$.
Our goal is to construct incentives that are robust with respect to the uncertainty induced by $B$.
More precisely, we are interested in incentives that work well for any present bias contained in $B$.
An alternative perspective is that we try to construct incentives which are not limited to a single person, but serve an entire population of individuals with different present bias values.
A simple instance of this problem in which a single task must be partitioned and stretched over a longer period of time has been studied by Kleinberg and Oren~\cite{KO}.
But like most research on incentivizing {\em heterogeneous} populations, see e.g.~\cite{OR3}, Kleinberg and Oren's results are restricted to a very specific setting.
They themselves suggest the design of more general incentives as a major research direction for the graphical~framework~\cite{KO}.

Using penalty fees as our incentive of choice and a fixed reward to keep people motivated, we present the first results in this area.
Our contribution is twofold.
On the one hand, we try to quantify the conceptual loss of efficiency caused by incomplete knowledge of $\beta$.
For this purpose we introduce a novel concept called {\em price of uncertainty}, which denotes the smallest ratio between the reward required by an incentive that accommodates all $\beta \in B$ and the reward required by an incentive designed for a specific $\beta \in B$.
We present an elegant algorithmic argument to prove that the price of uncertainty is at most $2$.
Remarkably, this bound holds true independent of the underlying graph $G$ and present bias set $B$.
To complement our result, we construct a family of graphs $G$ and present bias sets $B$ for which the price of uncertainty converges to a value strictly greater than~$1$.
On the other hand, we consider the computational problem of constructing penalty fees that work for all $\beta \in B$, but require as little reward as possible.
Drawing on the same algorithmic ideas we used to bound the price of uncertainty yields a polynomial time $2$-approximation.
Furthermore, we present a non-trivial proof to show that the decision version of the problem is contained in NP.
Since all hardness results of~\cite{AK2} also apply under uncertainty, we know that there is no $1{.}08192$-approximation unless ${\rm P = NP}$.

\subsection{Incentive Design for a Variable Present Bias}
In Section~\ref{sec:vpb} we generalize our notion of uncertainty to individuals whose present bias $\beta$ may change arbitrarily over time within the set $B$.
This model is inspired by work of Gravin et al.~\cite{GILP}, except that we do not rely on the assumption that $\beta$ is drawn independently from a fixed probability distribution.
Instead, our goal is to design penalty fees that work well for all possible sequences of $\beta$ over time.
We believe this to be an interesting extension of the fixed parameter case as the variability of $\beta$ may capture changes in a person's temporal discount function caused by unforeseen cognitive biases different from her present bias. 
As a result we obtain more robust penalty fees.

Again, our contribution is twofold.
On the one hand, we introduce the {\em price of variability} to quantify the conceptual loss of efficiency caused by unpredictable changes in $\beta$.
Similar to the price of uncertainty, we define this quantity to be the smallest ratio between the reward required by an incentive that accommodates all possible changes of $\beta \in B$ over time and the reward required by an incentive designed for a specific and fixed $\beta \in B$.
However, unlike the price of uncertainty, the price of variability has no constant upper bound.
Instead, the ratio seems closely related to the {\em range} $\tau = \max{B}/\min{B}$ of the set $B$.
By generalizing our algorithm from Section~\ref{sec:upb} we obtain an upper bound of $1 + \tau$ for the price of variability.
To complement this result, we construct a family of graphs $G$ for which the price of variability converges to $\tau/2$.
On the other hand, we consider the computational aspects of constructing penalty fees for a variable $\beta$.
As a result of the unbounded price of variability, we are not able to come up with a constant polynomial time approximation.
Instead, we obtain a ${(1 + \tau)}$-approximation.
However, by using a sophisticated reduction from VECTOR SCHEDULING, we prove that no efficient constant approximation is possible unless ${\rm NP = ZPP}$.
We conclude our work by studying a curious special case of variability in which individuals may temporarily lose their present bias.
For this scenario, which is characterized by the assumption that $1 \in B$, optimal penalty fees can be computed in polynomial time.

\section{The Model}
\label{sec:model}

In the following we introduce Kleinberg and Oren's graphical framework~\cite{KO}.
Let $G =(V,E)$ be a directed acyclic graph with $n$ nodes that models some long-term project.
The start and end states are denoted by the terminal nodes $s$ and~$t$.
Furthermore, each edge $e$ of $G$ corresponds to a specific task whose inured effort is captured by a non-negative cost $c(e)$.
To finish the project, a present biased agent must sequentially complete all tasks along a path from $s$ to~$t$.
However, instead of following a fixed path, the agent constructs her path dynamically according to the following simple procedure: 

When located at any node $v$ different from $t$, the agent tries to evaluate the minimum cost she needs to pay in order to reach $t$.
For this purpose she considers all outgoing edges $(v,w)$ of her current position $v$.
Because the tasks associated with these edges must be performed immediately, the agent assesses their cost correctly.
In contrast, all future tasks, i.e., tasks on a path from $v$ to $t$ not incident to $v$, are discounted by her present bias of $\beta \in (0,1]$.
As a result, we define her {\em perceived cost} for taking $(v,w)$ to be $d_{\beta}(v,w) = c(v,w) + \beta d(w)$, where $d(w)$ denotes the cost of a cheapest path from $w$ to $t$.
Furthermore, we define $d_{\beta}(v) = \min\{c(v,w) + \beta d(w) \mid (v,w) \in E\}$ to be the agent's {\em minimum perceived cost} at $v$.
Since the agent is oblivious to her own present bias, she only traverses edges $(v,w)$ for which $d_{\beta}(v,w) = d_{\beta}(v)$.
Ties are broken arbitrarily.
Once the agent reaches the next node, she reiterates this process.

To motivate the agent, a non-negative reward $r$ is placed at $t$.
Because the agent must reach $t$ before she can collect $r$, her {\em perceived reward} for reaching $t$ is $\beta r$ at each node different from $t$.
When located at $v \neq t$, the agent is only motivated to proceed if ${d_{\beta}(v) \leq \beta r}$.
Otherwise, if $d_{\beta}(v) > \beta r$, she quits.
We say that $G$ is {\em motivating}, if she does not quit while constructing her path from $s$ to~$t$.
Note that sometimes the agent can construct more than one path from $s$ to $t$ due to ties in the perceived cost of incident edges.
In this case, $G$ is considered motivating if she does not quit on any such path.

For the sake of a clear presentation, we will assume throughout this work that each node of $G$ is located on a path from $s$ to $t$.
This assumption is sensible because the agent can only visit nodes reachable from $s$.
Furthermore, she is not willing to enter nodes that do not lead to the reward at $t$.
Consequently, only nodes that are on a path from $s$ to $t$ are relevant to her behavior.
All nodes not satisfying this property can be removed from $G$ in a simple preprocessing step.

\subsection{Alice and Bob's Scenario}

To illustrate the model, we revisit Alice and Bob's scenario.
The task graph $G$ is depicted in Figure~\ref{fig:ex}.
Remember that $a = 1/2 - \varepsilon$ and $b = 1/2 + \varepsilon$ denote Alice and Bob's respective present bias.
For convenience let $0 < \varepsilon \leq 1/54$.
Furthermore, assume that a reward of $r = 27$ is awarded upon reaching~$t$.

We proceed to analyze Alice and Bob's walk through $G$.
At their initial position $s$ they must decide whether they move to $v_A$ or $v_B$.
For this purpose they try to find a path that minimizes the perceived cost.
As the more present biased person, Alice's favorite path is $s,v_{A},v_{AA},t$ with a perceived cost of $d_a(s) = d_a(s,v_A) = 8 - 14\varepsilon$.
By choice of $\varepsilon$ this cost is covered by her perceived reward ${ar = 27/2 - 27\varepsilon}$.
Consequently, she is motivated to traverse the first edge and moves to $v_A$.
A similar argument shows that Bob moves to $v_B$.
Once they reach their new nodes, Alice and Bob reevaluate plans.
From Alice's perspective $v_{A},v_{AA},t$ is still the cheapest path to $t$.
Bob, however, suddenly prefers $v_{B},v_{AB},t$ to his original plan.
Nevertheless, both of their perceived cost remains covered by their perceived reward and they move to $v_{AA}$ and $v_{AB}$ respectively.
At this point the only option is to take the direct edge to~$t$.
For Alice the perceived cost at $v_{AA}$ is sufficiently small to let her reach $t$.
In contrast, Bob's perceived cost~of ${d_b(v_{AB}) = 16}$ exceeds his perceived reward of $br = 27/2 + 27\varepsilon$ and he quits.

\subsection{Cost Configurations}

Bob's behavior in the previous example demonstrates how present biased decisions can deter people from reaching predefined goals.
To ensure an agent's success it is therefore sometimes necessary to implement external incentives such as penalty fees.
In the graphical model, penalty fees allow us to arbitrarily raise the cost of edges in $G$.
More formally, let $\tilde{c}$ be a so called {\em cost configuration}, which assigns a non-negative extra cost $\tilde{c}(e)$ to all edges $e$ of $G$.
The result is a new task graph $G_{\tilde{c}}$, whose edges $e$ have a cost of $c(e) + \tilde{c}(e)$.
A present biased agent navigates through $G_{\tilde{c}}$ according to the same rules applying in $G$.
We say that $\tilde{c}$ is motivating if and only if $G_{\tilde{c}}$ is.
To avoid ambiguity we annotate our notation whenever we consider a specific $\tilde{c}$, e.g., we write $d_{\tilde{c}}$ and $d_{\beta,\tilde{c}}$ instead of $d$ and $d_{\beta}$.

We conclude this section with a brief demonstration of the positive effects penalty fees can have in Alice and Bob's scenario.
Let $\tilde{c}$ be a cost configuration that assigns an extra cost of $\tilde{c}(v_{B},v_{AB}) = 1/2$ to $(v_{B},v_{AB})$ and $\tilde{c}(e) = 0$ to all other edges $e \neq (v_{B},v_{AB})$.
Note that $G$ and $G_{\tilde{c}}$ are identical task graphs except for the cost of $(v_{B},v_{AB})$.
Because Alice does not plan to take $(v_{B},v_{AB})$ on her way through $G$ and has even less reason to do so in $G_{\tilde{c}}$, we know that $\tilde{c}$ does not affect her behavior.
For similar reasons, $\tilde{c}$ does not affect Bob's choice to move to $v_{B}$.
However, once Bob has reached $v_{B}$ his perceived cost of the path $v_{B},v_{AB},t$ is $d_{b,\tilde{c}}(v_B, v_{AB}) = 19/2 + 16\varepsilon$, whereas his perceived cost of $v_{B},v_{BB},t$ is only $d_{b,\tilde{c}}(v_B, v_{BB}) = 19/2 + \varepsilon$.
Since the latter option appears to be cheaper and is covered by his perceived reward, Bob proceeds to $v_{BB}$ and then onward to $t$.
As a result $\tilde{c}$ yields a task graph that is motivating for Alice and Bob alike.
This is a considerable improvement to the original task graph.

\section{Uncertain Present Bias}
\label{sec:upb}

In this section we consider agents whose present bias $\beta$ is uncertain in the sense that our only information about $\beta$ is its membership in some set $B \subset (0,1]$.
We call $B$ the {\em present bias set}.
For technical reasons we assume that $B$ can be expressed as the union of constantly many closed subintervals from the set $(0,1]$.
This way the intersection of $B$ with a closed interval is either empty or contains an efficiently computable minimal and maximal element.
To measure the degree of uncertainty induced by $B$, we define the range of $B$ as $\tau = \max{B}/\min{B}$.

\subsection{A Decision Problem}

Our goal is to construct a cost configuration $\tilde{c}$ that is motivating for all $\beta \in B$, but requires as little reward as possible. 
To assess the complexity of this task, let UNCERTAIN PRESENT BIAS (UPB) be the following decision problem:
\begin{definition}[UPB]
	Given a task graph $G$, present bias set $B$ and reward $r > 0$, decide whether a cost configuration $\tilde{c}$ motivating for all $\beta \in B$ exists.
\end{definition}
If $\tau = 1$, i.e., $B$ only contains a single present bias parameter, UPB is identical to the decision problem MOTIVATING COST CONFIGURATION (MCC) studied in \cite{AK2}.
Since MCC is NP-complete, UPB must be NP-hard.
But unlike MCC it is not immediately clear if UPB is also contained in NP.
The reason is that proving ${\rm MCC} \in {\rm NP}$ only requires to verify whether a given cost configuration is motivating for a single value of $\beta$; a property that can be checked in polynomial time~\cite{AK}.
However, proving ${\rm UPB} \in {\rm NP}$ requires to verify whether a given cost configuration is motivating for all $\beta \in B$.
Taking into account that $B$ may very well be an infinite set, it becomes clear that we cannot check all values of $\beta$ individually.
Interestingly, we do not have to; checking a finite subset $B' \subseteq B$ of size $\mathcal{O}(n^2)$ turns out to be sufficient.

\begin{proposition}\label{prop:mccunp}
	For any task graph $G$, reward $r$ and present bias set $B$ a finite subset $B' \subseteq B$ of size $\mathcal{O}(n^2)$ exists such that $G$ is motivating for all $\beta \in B$ if it is motivating for all $\beta \in B'$.
\end{proposition}

\begin{proof}
	Our proof consists of two steps. 
	First, we determine for what $\beta \in [0,1]$ the {\em preference profile} $p_{\beta}(v) = \{(v,w') \mid d_{\beta}(v,w') = d_{\beta}(v)\}$, i.e., the set of edges an agent with present bias $\beta$ is inclined to traverse, contains a given edge $(v,w)$.
	Let $B_{v,w}$ be the set of all $\beta$ satisfying this property.
	More formally we define $B_{v,w}$ as the intersection $\{\beta \mid (v,w) \in p_{\beta}(v)\} \cap [0,1]$. 
	As we are going to show, $B_{v,w}$ is a closed, possibly empty, subinterval of $[0,1]$. 
	For the second step we define $B'$ to be the set $\{\min(B_e \cap B) \mid B_e \cap B \neq \emptyset\}$.
	Note that by our assumption on $B$ the minimum of $B_e \cap B$ is guaranteed to exist as long as $B_e \cap B \neq \emptyset$. 
	Furthermore, $B'$ contains at most $|E|$ elements and therefore has a size of $\mathcal{O}(n^2)$. 
	Using a proof by contradiction, we argue that $G$ is motivating for all $\beta \in B$ if it is motivating for all $\beta \in B'$.
	This completes the proof.
	
	To see that $B_{v,w}$ is a closed subinterval of $[0,1]$, it is instructive to observe how the perceived cost of $(v,w)$ changes with respect to $\beta$.
	For this purpose, let ${\ell_{v,w}(\beta) = c(v,w) + \beta d(w)}$ denote the perceived cost $d_{\beta}(v,w)$ with respect to $\beta$.
	Clearly, the function $\ell_{v,w}$ is linear and so are the functions $\ell_{v,w'}$ associated with the incident edges of $(v,w)$.
	Together, these functions yield the line arrangement $L_v = \{\ell_{v,w'} \mid (v,w') \in E\}$.
	By definition of the agent's preference profile, $p_{\beta}(v)$ contains $(v,w)$ if and only if $d_{\beta}(v,w) \leq d_{\beta}(v,w')$ for all $(v,w') \in E$.
	The values of $\beta \in [0,1]$ that satisfy this property are exactly those for which the line $\ell_{v,w}$ is on the lower envelope of the line arrangement $L_v$.
	From the basic structure of line arrangements such as $L_v$ we can immediately conclude that $B_{v,w}$ must be a closed subinterval of $[0,1]$, see e.g. \cite{E}.

	To conclude the proof it remains to show that $G$ is motivating for each $\beta \in B$ if it is motivating for each $\beta \in B'$.
	For the sake of contradiction assume that $G$ is motivating for each $\beta \in B'$ but not for some $b \in B$.
	As a result there must exist a path $P$ from $s$ to $t$ along which an agent with present bias $b$ may walk and which contains an edge $(v,w)$ such that $d_{b}(v,w)/b > r$.
	Let $B_P = \bigcap_{e \in P} B_{e}$ be the set of all present bias parameters for which agents can construct $P$.
	From this definition it is immediately apparent that $b$ is contained $B_P$ and therefore $B_P \cap B \neq \emptyset$.
	We now consider the structure of $B_P$.
	Since each set $B_{e}$ associated with an $e \in P$ is a closed interval, so is $B_P$.
	In particular, as $B_P \cap B$ is not empty, one of these sets $B_{e}$ must satisfy $\min(B_{e} \cap B) = \min(B_P \cap B)$.
	Let $a = \min(B_{e} \cap B)$ denote the corresponding present bias value.
	By definition of $a$ it holds true that $a \in B'$ and $a \leq b$.
	Moreover, because $a \in B_P$, we know that an agent with present bias $a$ may very well construct the path $P$.
	However, once this agent reaches $(v,w)$, her perceived cost exceeds her perceived reward as
	\[\frac{d_{a}(v,w)}{a} = \frac{c(v,w)}{a} + d(w) \geq \frac{c(v,w)}{b} + d(w) = \frac{d_{b}(v,w)}{b} > r.\]
	Consequently, $G$ is not motivating for $a$.
	However, this is a contradiction to the assumption that $G$ is motivating for all $\beta \in B'$.
\end{proof}

Note that the proof of Proposition~\ref{prop:mccunp} does not only imply that $B'$ exists, but also that it can be constructed in polynomial time.
We may therefore conclude that UPB is indeed contained in NP.

\begin{corollary}\label{coll:upbnp}
	UPB is NP-complete.
\end{corollary}

\subsection{The Price of Uncertainty}

Since UPB is NP-complete, it makes sense to consider the corresponding optimization problem UPB-OPT.
For this purpose, let $r(G,B)$ be the infimum over all rewards admitting a cost configuration motivating for all $\beta \in B$ and define:
\begin{definition}[UPB-OPT]
	Given a task graph $G$ and present bias set $B$, determine $r(G,B)$.
\end{definition}
Clearly, UPB-OPT must be at least as hard as the optimization version of MCC.
Consequently, we know that UPB has no PTAS and is NP-hard to approximate within a ratio less than $1.08192$~\cite{AK2}.
But does the transition from a certain to an uncertain $\beta$ reduce approximability?

Setting complexity theoretic considerations aside for a moment, an even more general question arises:
How does the transition from a certain to an uncertain $\beta$ affect the efficiency of cost configurations assuming unlimited computational resources?
To quantify this conceptual difference in efficiency, we look at the smallest ratio between optimal cost configurations motivating for all $\beta \in B$ and optimal cost configurations motivating for a specific $\beta \in B$.
We call this ratio the {\em price of uncertainty}.
\begin{definition}[Price of Uncertainty]\label{def:pou}
	Given a task graph $G$ and a present bias set $B$, the price of uncertainty is defined as $r(G,B)/\sup\{r(G,\{\beta\}) \mid \beta \in B\}$.
\end{definition}

Let us illustrate the price of uncertainty by going back to Alice and Bob's scenario and assume that $B = \{a,b\}$ with $a = 1/2 - \varepsilon$ and $1/2 + \varepsilon$.
In other words, the agent either behaves like Alice or she behaves like Bob, but we do not know which.
It is easy to see that in either case the agent minimizes her maximum perceived cost on the way from $s$ to $t$ by taking the path $P = s,v_B,v_{BB},t$.
This minmax cost, which is either $d_a(v_B, v_{BB}) = 19/2 - \varepsilon$ or $d_b(v_B, v_{BB}) = 19/2 + \varepsilon$, provides two lower bounds for the necessary reward when divided by the respective present bias.
More formally, it holds true that ${r(G,\{a\}) \geq (19/2 - \varepsilon)/(1/2 - \varepsilon)}$ and ${r(G,\{b\}) \geq (19/2 + \varepsilon)/(1/2 + \varepsilon)}$.
However, as we have seen in Section~\ref{sec:model}, neither Alice nor Bob are willing to follow $P$ without external incentives.
To discourage the agent from leaving $P$, we assign an extra cost of $\tilde{c}(s,v_A) = 5\varepsilon$ to $(s,v_A)$, $\tilde{c}(v_B,v_{AB}) = 1/2 + 16\varepsilon$ to $(v_B,v_{AB})$ and $\tilde{c}(e) = 0$ otherwise.
This extra cost does not affect the agent's maximum perceived cost along $P$, which she still experiences at $(v_B,v_{BB})$.
As a result, our bounds for $r(G,\{a\})$ and $r(G,\{b\})$ are tight and we get ${\sup\{r(G,\{\beta\}) \mid \beta \in B\} = r(G,\{a\})}$.
Moreover, because we have used the same cost configuration $\tilde{c}$ to derive $r(G,\{a\})$ and $r(G,\{b\})$, it must hold true that $r(G,B) = \sup\{r(G,\{\beta\}) \mid \beta \in B\}$, implying that the price of uncertainty in Alice and Bob's scenario is $1$.

\subsection{Bounding the Price of Uncertainty}

As Alice and Bob's scenario demonstrates, cost configurations designed for an uncertain $\beta$ are not necessarily less efficient than those designed for a specific~$\beta$.
Therefore one might wonder whether scenarios exist in which a real loss of efficiency is bound to occur, i.e., can the price of uncertainty be greater than $1$?
The following proposition shows that such scenarios indeed exist.

\begin{proposition}\label{prop:lpu}
	There exists a family of task graphs and present bias sets for which the price of uncertainty converges to $1.1$.
\end{proposition}

\begin{proof}
	Let $0 < a \leq 3/8$ be some present bias such that $4/a$ is integral and consider the task $G$ consisting of a directed path $v_0,v_1,\ldots,v_{12 + 4/a}$.
	We call this path the {\em main path} and charge a cost of $2$ on its first edge while all other edges of the path have a cost of $1$.
	In addition to the main path, we introduce a {\em shortcut} from $v_1$ to $v_{12 + 4/a}$ along an intermediate node $w$.
	The cost of the two edges $(v_1,w)$ and $(w,v_{12 + 4/a})$ is $4$ and $6 + 3/a$ respectively.
	For a convenience let $s = v_0$ and $t = v_{12 + 4/a}$.
	Figure~\ref{fig:lpu} shows a sketch of $G$.
	
	Assume that $G$ is traversed by an agent whose present bias is either $a$ or $b=1/2$, but we do not know which, i.e., $B = \{a, 1/2\}$.
	Our goal is to construct two cost configurations $\tilde{a}$ and $\tilde{b}$ such that $G_{\tilde{a}}$ and $G_{\tilde{b}}$ are motivating for a reward of $10 + 5/a$ and the respective present bias.
	This implies $\sup\{r(G,\{a\}),r(G,\{1/2\})\} \leq 10 + 5/a$.
	We then argue that a reward less than $9 + 11/(2a)$ is not sufficient to motivate both types of agents simultaneously, i.e., $r(G,B) \geq 9 + 11/(2a)$.
	As $(9 + 11/(2a))/(10 + 5/a)$ converges to $11/10$ for $a \to 0$, this establishes the proposition.
	
	We begin with~$\tilde{a}$.
	For this purpose let $\tilde{a}$ assign no extra cost at all, i.e., $\tilde{a}(e) = 0$ for all edges~$e$.
	Furthermore, assume the agent has a present bias of $a$ and a reward of $10 + 5/a$ is placed at $t$.
	When located at $s$, the agent's only choice is $(s,v_1)$.
	If she plans to take the shortcut afterwards, her perceived cost of $(s,v_1)$ is at most $d_{a,\tilde{a}}(s,v_1) \leq 10a + 5$.
	As this matches her perceived reward, she proceeds to $v_1$ where she faces two options:
	The first one is to take the shortcut for a perceived cost of $d_{a,\tilde{a}}(v_1,w) = 6a + 7$.
	Considering that we chose $a$ to satisfy $a < 1/2$, it follows that $d_{a,\tilde{a}}(v_1,w) > 10a + 5$ and therefore the shortcut is not motivating.
	The second option is to take the main path along $11 + 4/a$ edges of cost $1$, resulting in a perceived cost of $d_{a,\tilde{a}}(v_1,v_2) = 10a + 5$.
	Similar to the situation at $s$, this cost matches the perceived reward and she proceeds to $v_2$.
	Since all remaining edges $(v_{i},v_{i+1})$ have a perceived cost less than $d_{a,\tilde{a}}(v_1,v_2)$, the agent eventually reaches $t$ and we conclude that $G_{\tilde{a}}$ is motivating.
	
	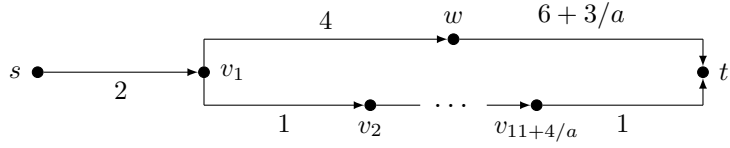
\begin{figure}[t]
		\center
		\begin{tikzpicture}[scale=0.875, nst/.style={draw,circle,fill=black,minimum size=4pt,inner sep=0pt]}, est/.style={draw,>=latex,->}]
			
			\node[nst] (v1) at (0,0.5) [label=left:\small{$s$}] {};
			\node[nst] (v2) at (2.5,0.5) [label=right:\small{$v_1$}] {};
			\node[nst] (v3) at (5,0) [label=below:\small{$v_2$}] {};
			\node[nst] (v4) at (7.5,0) [label=below:\small{$v_{11 + 4/a}$}] {};
			\node[nst] (v5) at (10,0.5) [label=right:\small{$t$}] {};
			\node[nst] (v6) at (6.25,1) [label=above:\small{$w$}] {};
			
			\node at (6.25,0) {$\dots$};			
			
			\path (v1) edge[est] node [below] {\small{$2$}} (v2)
			(v2) edge (2.5,0)
			(2.5,0) edge[est] node [below] {\small{$1$}} (v3)
			(v3) edge (5.75,0)
			(6.75,0) edge[est] (v4)
			(v4) edge node [below] {\small{$1$}} (10,0)
			(10,0) edge[est] (v5)
			(v2) edge (2.5,1)
			(2.5,1) edge[est] node [above] {\small{$4$}} (v6)
			(v6) edge node [above] {\small{$6 + 3/a$}} (10,1)
			(10,1) edge[est] (v5);
			\end{tikzpicture}
		\caption{Task graph with a price of uncertainty of $(9 + 11/(2a))/(10 + 5/a)$}\label{fig:lpu}
	\end{figure}	
	
	We continue to construct $\tilde{b}$ by setting $\tilde{b}(v_1,w) = 1/(2a)$ and $\tilde{b}(e) = 0$ otherwise.
	In contrast to the previous scenario assume the agent has a present bias of $b$.
	The reward is still $10 + 5/a$, but its perceived value has changed to $5 + 5/(2a)$.
	When located at the initial node $s$, the agent's perceived cost is at most $d_{b,\tilde{b}}(s,v_1) \leq 7 + 7/(4a)$ if she plans to take the shortcut afterwards.
	By choice of $a$ it holds true that $a < 3/8$ and we get $d_{b,\tilde{b}}(s,v_1) \leq 5 + 5/(2a)$.
	Consequently, the agent is motivated to proceed to $v_1$.
	At this point, she has to choose between the shortcut and the main path.
	Her perceived cost of the former is $d_{b,\tilde{b}}(v_1,w) = 7 + 2/a$, whereas the latter has a perceived cost of  $d_{b,\tilde{b}}(v_1,v_2) = 6 + 2/a$.
	Clearly, the main path is her preferred choice and since $a < 1/2$, it is also a motivating one.
	Because all remaining edges $(v_i,v_{i+1})$ have a perceived cost less than $d_{b,\tilde{a}}(v_1,v_2)$, it follows that $G_{\tilde{b}}$ is motivating.
	
	It remains to show that no cost configuration can be motivating for both types of agents at the same time if the reward is less than $9 + 11/(2a)$.
	For the sake of contradiction assume such a cost configuration $\tilde{c}$ exists.
	Note that an agent with present bias $b$ must not enter the shortcut as her perceived cost $d_{b,\tilde{c}}(w,t) = 6 + 3/a$ exceeds her perceived reward of $9/2 + 11/(4a)$ for any $a > 0$.
	However, if we do not assign extra cost to $G$, such an agent prefers the shortcut to the main path when located at $v_1$.
	The difference in perceived cost is $d_{b}(v_1,v_2) - d_{b}(v_1,w) = 1/(2a) - 1$.
	Consequently, $\tilde{c}$ must assign extra cost greater than $1/(2a) - 1$ to the shortcut.
	Next consider an agent with a present bias of $a$ located at $s$.
	At this point her perceived cost for taking the shortcut is greater than $9a + 11/2$, due to the extra cost assigned by $\tilde{c}$.
	Note that this cost exceeds her perceived reward.
	Her other option is to plan along the main path.
	In this case her perceived cost is $11a + 6$.
	Clearly, this is even more expensive contradicting the assumption that $\tilde{c}$ is motivating for both types of agents.
\end{proof}

As the price of uncertainty can be strictly greater than $1$, the question for an upper bound arises.
Ideally, we would like to design a cost configuration $\tilde{c}$ motivating for all $\beta \in B$ assuming the reward is set to $\varrho r(G,\{b\})$ for some constant factor $\varrho > 1$ and $b = \min B$.
Clearly, the existence of such a $\tilde{c}$ would imply a constant bound of $\varrho$ for the price of uncertainty independent of $G$ and $B$.
Using a generalized version of the approximation algorithm we proposed in~\cite{AK2}, it is indeed possible to construct a $\tilde{c}$ with the desired property for $\varrho = 2$.

\begin{algorithm}[t]
	\caption{\sc UncertainPresentBiasApprox} \label{alg:uaprox}
	\SetKw{KwSuch}{such that}
	\SetKw{KwAnd}{and}
    $b \gets \min B$; $P \gets$ minmax path from $s$ to $t$ w.r.t $d_{b}(e)$; $\alpha \gets \max\{d_{b}(e) \mid e \in P\}$\;
    \lForEach{$v \in V\setminus\{t\}$}{$\varsigma(v) \gets \text{successor of } v \text{ on a cheapest path from } v \text{ to } t$}
    $T = \{(v,\varsigma(v)) \mid v \in V \setminus \{t\}\}$\;
    \lForEach{$e \in E$}{$\tilde{c}(e) \gets 0$}
	\lForEach{$e \in E \setminus (P \cup T)$}{$\tilde{c}(e) \gets 2\alpha/b + 1$}\label{ln:c1}
	\ForEach{$(v,w) \in T$ \KwSuch $v \in P$ \KwAnd $w \notin P$}{
		$P' \gets v,\varsigma(v),\varsigma(\varsigma(v)),\ldots,t$\;
		$u \gets \text{first node of } P' \text{ different from } v \text{ that is also a node of } P$\;
		$\tilde{c}(v,w) \gets \text{cost of most expensive edge of } P' \text{ between } v \text{ and } u$\;\label{ln:c2}
	}
	\Return $\tilde{c}$\;
\end{algorithm}

The main idea of {\sc UncertainPresentBiasApprox} is simple:
First, the algorithm computes a value $\alpha$ such that $\alpha/b$ is a lower bound on the reward necessary for agents with present bias $b$, i.e., $r(G,\{b\}) \geq \alpha/b$.
In particular, this bound implies $\sup\{r(G,\{\beta\}) \mid \beta \in B\} \geq \alpha/b$.
Next the algorithm constructs a $\tilde{c}$ such that a reward of $2\alpha/b$ is sufficiently motivating for all $\beta \in B$, i.e., $r(G,B) \leq 2\alpha/b$.
As a result the price of uncertainty can be at most $2$.
In the following we try to convey the intuition behind the algorithm in more detail.

We begin with the computation of $\alpha$.
For this purpose let $P$ be a path minimizing the maximum cost an agent with present bias $b$ perceives on her way from $s$ to $t$.
We call $P$ a {\em minmax path} and define $\alpha = \max\{d_{b}(e) \mid e \in P\}$ to be the maximum perceived edge cost of $P$.
Since cost configurations cannot decrease edge cost, it should be clear that $\alpha$ is a valid lower bound on the reward required for the present bias $b$, i.e., $r(G,\{b\}) \geq \alpha/b$.

We proceed with $\tilde{c}$.
The goal is to assign extra cost in such a way that any agent with a present bias $\beta \in B$ traverses only two kinds of edges.
The first kind of edges are those on $P$.
It is instructive to note that each such edge $(v,w) \in P$ is motivating for a reward of $\alpha/b$ if $\beta \geq b$.
The reason is that
\[d_{\beta}(v,w) = \beta\Bigl(\frac{c(v,w)}{\beta} + d(v,w)\Bigr) \leq \beta\Bigl(\frac{c(v,w)}{b} + d(v,w)\Bigr) = \beta\frac{d_{b}(v,w)}{b} = \beta \frac{\alpha}{b}.\]
In particular, $P$ is motivating for each present bias $\beta \in B$.
The second kind of edges are on cheapest paths to $t$.
To identify these edges, the algorithm assigns a distinct successor $\varsigma(v)$ to each node $v \in V \setminus \{t\}$ such that $(v,\varsigma(v))$ is the initial edge of a cheapest path from $v$ to $t$.
Since we assume $t$ to be reachable from all other nodes of $G$ at least one suitable successor must exist.
By definition of $\varsigma$, we know that $P' = v,\varsigma(v),\varsigma(\varsigma(v)),\ldots,t$ is a cheapest path from $v$ to $t$.
We call $P'$ the {\em $\varsigma$-path\/} of $v$ and $T = \{(v,\varsigma(v)) \mid v \in V \setminus \{t\}\}$ a {\em cheapest path tree\/}.

Remember that we try to keep agents on the edges of $P$ and $T$.
For this purpose, we assign an extra cost of $\tilde{c}(e) = 2\alpha/b + 1$ to all other edges.
This raises their perceived cost to at least $2\alpha/b + 1$; a price no agent is willing to pay for a perceived reward of $\beta2\alpha/b$.
However, since we have not assigned any extra cost to $T$ so far, the perceived cost of edges in $P$ and $T$ is unaffected by the current~$\tilde{c}$.
In particular, all edges of $P$ are still motivating for a reward of $\alpha/b$ and any present bias $\beta \in B$.
To keep agents from entering costly $\varsigma$-paths $P' = v,\varsigma(v),\varsigma(\varsigma(v)),\ldots,t$, we assign an extra cost to the out-edges $(v,\varsigma(v))$ of $P$, i.e., $v \in P$ but $\varsigma(v) \notin P$.
The extra cost $\tilde{c}(v,\varsigma(v))$ is chosen to match the cost of a most expensive edge on $P'$ between $v$ and the next intersection of $P'$ and~$P$.
It is easy to see that the resulting $\tilde{c}$ can no more than double the perceived cost of any edge in $P$, see the proof of Theorem~\ref{thm:upu} for a precise argument.
Furthermore, the perceived cost of any out-edge $(v,\varsigma(v))$ of $P$ is either high enough to keep agents on $P$ or they do not encounter edges exceeding the perceived cost of $(v,\varsigma(v))$ until they reenter $P$.
We conclude that a reward of $2\alpha/b$ is sufficiently motivating, leading us to one of the central results of our work.

\begin{theorem}\label{thm:upu}
	The price of uncertainty is at most $2$.
\end{theorem}

\begin{proof}
	Let $\beta$ be an arbitrary present bias from $B$.
	According to our considerations from Section~\ref{sec:upb} it is evident that the theorem holds true if {\sc UncertainPresentBiasApprox} yields a cost configuration $\tilde{c}$ such that a reward of $2\alpha/b$ is motivating for $\beta$.
	In other words, we need to show that no agent with present bias $\beta$ can reach a node $v$ where her perceived cost $d_{\beta,\tilde{c}}(v)$ exceeds her perceived reward of~$\beta2\alpha/b$.
	For this purpose we make a case distinction on $v$.
	
	First, assume that $v$ is a node of $P$ and let $w$ be the direct successor of $v$ on $P$. 
	In our previous discussion of {\sc UncertainPresentBiasApprox} we have already argued that $(v,w)$ is motivating for a reward of $\alpha/b$ if we ignore extra cost, i.e., $d_{\beta}(v,w) \leq \beta\alpha/b$.
	Our goal is to show that $\tilde{c}$ can at most double the perceived cost of $(v,w)$ in the sense that $d_{\beta,\tilde{c}}(v,w) \leq 2 d_{\beta}(v,w)$.
	Clearly, this directly implies the desired bound $d_{\beta,\tilde{c}}(v) \leq d_{\beta,\tilde{c}}(v,w) \leq \beta2\alpha/b$ on the perceived cost of $v$.
	To prove our claim, let $P'$ be the $\varsigma$-path of~$w$.
	Remember {\sc UncertainPresentBiasApprox} does not raise the cost of $(v,w)$.
	Moreover, the algorithm only assigns extra cost to edges $(w',\varsigma(w'))$ of $P'$ that are out-edges of $P$, i.e., $w' \in P$ but $\varsigma(w') \notin P$.
	Consequently, at most one edge of $P'$ can charge extra cost between any two consecutive intersections of $P$ and $P'$.
	Because this extra cost is equal to the cost of an edge between the two intersections, each edge of $P'$ can contribute at most once to the total extra cost of $P'$.
	As a result, we know that the extra cost of $P'$ is less or equal to its original cost.
	Together with the fact $P'$ is a cheapest path from $w$ to $t$ we conclude that
	
	\begin{align*}
		d_{\beta,\tilde{c}}(v,w) 	&= c(v,w) + \beta d_{\tilde{c}}(w) \leq c(v,w) + \beta \sum_{e \in P'}\bigl(c(e) + \tilde{c}(e)\bigr)\\
								&\leq c(v,w) + \beta \sum_{e \in P'} 2c(e) = c(v,w) + 2\beta d(w) \leq 2 d_{\beta}(v,w).
	\end{align*}
	
	Next assume that $v$ is not on $P$ and let $v'$ be the last node of $P$ the agent visited before reaching~$v$.
	Furthermore, let $P_v = v,w,\ldots,t$ and $P_{v'} = v',w',\ldots,t$ be the two paths the agent plans to take when located at $v$ and $v'$ respectively.
	Because the agent is willing to traverse $(v',w')$, we know $d_{\beta,\tilde{c}}(v',w') \leq \beta2\alpha/b$.
	To obtain the desired bound $d_{\beta,\tilde{c}}(v) \leq d_{\beta,\tilde{c}}(v,\varsigma(v)) \leq \beta2\alpha/b$, all that remains to be shown is $d_{\beta,\tilde{c}}(v,w) \leq d_{\beta,\tilde{c}}(v',w')$.
	At this point it is helpful to recall that the agent never plans to take an edge that is neither part of $P$ or $T$.
	The reason is that those edges have an extra cost of $2\alpha/b + 1$, which certainly exceeds the perceived reward.
	As a result, we know that the first edges of $P_{v'}$, i.e., the edges from $v'$ to $w$, are all edges of $T$.
	Let $Q$, be the subpath of $P_{v'}$ that goes from $w'$ to~$w$.
	By definition of $P_{v'}$ we know that $Q$ is part of a cheapest path from $w'$ to $t$ with respect to $\tilde{c}$ and it becomes apparent that
	\[d_{\tilde{c}}(w) \leq \sum_{e \in Q}\bigl(c(e) + \tilde{c}(e)\bigr) + d_{\tilde{c}}(w) = d_{\tilde{c}}(w').\]
	Moreover, by design of $\tilde{c}$ we know that $c(v,w) \leq \tilde{c}(v',w')$ and $\tilde{c}(v,w) = 0$.
	Combining these inequalities immediately yields
	\[d_{\beta,\tilde{c}}(v,w) = c(v,w) + \beta d_{\tilde{c}}(w) \leq c(v',w') + \tilde{c}(v',w') + \beta d_{\tilde{c}}(w') = d_{\beta,\tilde{c}}(v',w').\]
	This completes the proof. 
\end{proof}

It is interesting to note that {\sc UncertainPresentBiasApprox} can be executed in polynomial time.
Furthermore, the proof of Theorem~\ref{thm:upu} argues that $\alpha/b \leq r(G,B) \leq 2\alpha/b$.
As a result we have also found an efficient constant factor approximation of UPB-OPT.

\begin{corollary}\label{coll:uaprox}
	UPB-OPT admits a polynomial time $2$-approximation.
\end{corollary}

\section{Variable Present Bias}
\label{sec:vpb}

So far we have considered agents with an unknown but fixed present bias.
We now generalize this model to agents whose $\beta$ may vary arbitrarily within $B$ as they progress through $G$.
It is convenient to think of $\beta$ as a {\em present bias configuration}, i.e., an assignment of present bias values $\beta(v) \in B$ to the nodes $v$ of~$G$.
Whenever the agent reaches a node $v$, she acts according to the current present bias value $\beta(v)$.
We say that $G$ is motivating with respect to a present bias configuration $\beta$ if and only if the agent does not quit on a walk from $s$ to $t$.

To illustrate the consequences of a variable present bias we revisit Alice and Bob's scenario once more.
Recall that the agent in this scenario is either like Alice with a present bias of $a = 1/2 - \varepsilon$ or like Bob with a present bias of $b = 1/2 + \varepsilon$, i.e., $B = \{a,b\}$.
But while she had to commit to one present bias before, she is now free to change between $a$ and $b$.
For instance, her present bias could be $b$ at $s$ and $v_B$, but $a$ otherwise, i.e., $\beta(v) = b$ for $v \in \{s,v_B\}$ and $\beta(v) = a$ for $v \in V \setminus \{s,v_B\}$.
In this case she walks along the same path Bob would take, i.e., $s,v_B,v_{AB},t$.
However, there is a subtle difference.
At $v_{AB}$ the agent behaves like Alice and needs strictly more reward than Bob to remain motivated while traversing $(v_{AB},t)$.
Under closer examination, which we will not go into detail here, it is in fact easy to see that the variability of $\beta$ makes our agent more expensive to motivate than any agent with a fixed present bias from~$B$.

\subsection{Computational Consideration}

Let $G$ be an arbitrary task graph and $B$ a suitable present bias set.
We want to construct a cost configuration $\tilde{c}$ that is motivating for all present bias configuration $\beta \in B^V$, but requires as little reward as possible.
Using arguments similar to those of Section~\ref{sec:upb}, the computational challenges of this task are readily apparent.
In particular, the corresponding decision problem VARIABLE PRESENT BIAS (VPB) is equivalent to MCC whenever $B$ only contains a single element.

\begin{definition}[VPB]\label{def:mccv}
	Given a task graph $G$, present bias set $B$ and reward $r > 0$, decide whether a cost configuration $\tilde{c}$ motivating for all $\beta \in B^V$ exists.
\end{definition}

Because MCC is NP-complete~\cite{AK2}, it immediately follows that VPB is NP-hard.
However, proving that ${\rm UPB} \in {\rm NP}$ requires a more careful argument.

\begin{proposition}\label{prop:mccvnp}
  	VPB is contained in NP.
\end{proposition}

\begin{proof}
	Let $B$ an arbitrary present bias set given as input and $G$ the task graph resulting from some cost configuration.
	To prove that VPB is contained in NP we need to find an efficient way of confirming whether $G$ is motivating for all $\beta \in B^V$.
	Our strategy for this task consists of three steps:
	(a)~Construct the set $E'$ containing all edges $(v,w)$ minimizing the agent's perceived cost at $v$ for some present bias $\beta(v) \in B$.
	(b) Compute the set $V'$ containing all nodes reachable from $s$ via edges of $E'$.
	(c) For each node $v \in V'$ check that the agent remains motivated for all possible present bias values $\beta(v) \in B$.
	To complete the proof we argue that each step can be completed in polynomial time and that (c) holds true if and only if $G$ is motivating for all $\beta \in B^V$.
	
	We begin with step (a).
	Observe that $E'$ only contains edges $(v,w)$ for which a $\beta(v) \in B$ exists such that $d_{\beta(v)}(v,w) \leq d_{\beta(v)}(v,w')$ for all $(v,w') \in E$.
	From the proof of Proposition~\ref{prop:mccunp} we already know that the present bias parameters $\beta' \in [0,1]$ satisfying $d_{\beta'}(v,w) \leq d_{\beta'}(v,w')$ form a closed subinterval $B_{v,w} \subseteq [0,1]$.
	Moreover, the end points of this interval can be computed in polynomial time.
	Together with our assumptions on $B$, we may assume that $B_{v,w} \cap B \neq \emptyset$ can be determined efficiently and therefore $E'$ can be constructed in polynomial time.
	
	We continue with (b).
	Given the set $E'$, it is trivial to construct $V'$ in polynomial time.
	Furthermore, $V'$ contains exactly those nodes an agent with variable present bias can reach if the reward is sufficiently large.
	To see this, let $P$ be a path from $s$ to $v$ only using edges of~$E'$.
	If we go through the edges $(v',w')$ of $P$ and choose $\beta(v')$ as an element of the non-empty intersection $B_{v',w'} \cap B$, we obtain a valid present bias configuration that may lead the agent onto $v$ for a sufficiently large reward.
	Conversely, we can find no such present bias configuration for nodes $v \notin V'$.
	The reason is that all path $P$ to $v$ must contain at least one edge $(v,w) \notin E'$.
	But by definition of $E'$ the agent does not traverse $(v,w)$ for any $\beta(v) \in B$.
	
	We conclude with (c).
	As a result of (b), we know that the agent can never reach nodes outside of~$V'$.
	However, she can reach each $v \in V'$, unless she quits at some other node of $V'$ before she gets to~$v$.
	As a result, we know that $G$ is motivating for all $\beta \in B^V$ if and only if the agent never quits when located at any of the nodes $v \in V'$.
	Conveniently, we do not need to check the latter condition for all $\beta(v) \in B$.
	Instead, we only need to check for $\beta(v) = b$, with $b$ denoting the minimum of~$B$.
	The reason is that the minimum reward that is motivating in the case of $\beta(v) = b$ is greater or equal to the reward required by any other present bias $\beta(v) \geq b$, as the following inequality demonstrates
	\begin{align*}
		\frac{d_{b}(v)}{b}	&= \min\Bigl\{\frac{c(v,w)}{b} + d(w) \Bigm| (v,w) \in E \Bigr\}\\
							&\geq \min\Bigl\{\frac{c(v,w)}{\beta(v)} + d(w) \Bigm| (v,w) \in E \Bigr\} = \frac{d_{\beta(v)}(v)}{\beta(v)}.
	\end{align*}
	Therefore (c) can be decided in polynomial time.
	This completes the proof.
\end{proof}

As a result of Proposition~\ref{prop:mccvnp}, we may conclude that VPB is NP-complete.

\begin{corollary}
	VPB is NP-complete.
\end{corollary}

Since it is NP-hard to find optimal cost configurations for general~$B$, we turn to an optimization version of the problem.
Assuming that $r(G,B^V)$ denotes the infimum over all rewards admitting a cost configuration $\tilde{c}$ motivating for all $\beta \in B^V$, we define VPB-OPT as:
 
\begin{definition}[VPB-OPT]
	Given a task graph $G$ and present bias set $B$, determine $r(G,B^V)$.
\end{definition}

Interestingly, approximating VPB-OPT seems to be much harder than UPB-OPT.
The reason why the $2$-approximation for UPB-OPT, i.e., {\sc UncertainPresentBiasApprox}, does not work anymore is simple.
Recall that the cost configuration $\tilde{c}$ returned by the algorithm lets the agent take shortcuts along cheapest paths to $t$.
To ensure that these shortcuts do not become too expensive, $\tilde{c}$ assigns extra cost to their initial edge.
This way the perceived cost within a shortcut should not be greater than that for entering.
As long as the present bias is fixed, this works fine.
However, if the present bias can change, the agent may become more biased within a shortcut and require higher rewards to stay motivated.
One way to fix this problem is to let the assigned extra cost depend on $\tau$, i.e., the range of $B$.
More precisely, we multiply the cost assigned in line~\ref{ln:c2} of Algorithm~\ref{alg:uaprox} by $\tau$ and change line~\ref{ln:c1} to assign a cost of $\tilde{c}(e) = (1+\tau)\alpha/b + 1$.
As a result we obtain a new algorithm {\sc VariablePresentBiasApprox} with an approximation ration of $1 + \tau$.

\begin{theorem}\label{thm:vaprox}
	VPB-OPT admits a polynomial time $(1+\tau)$-approximation.
\end{theorem}

\begin{proof}
	Let $b = \min B$.
	From the analysis of {\sc UncertainPresentBiasApprox} in Section~\ref{sec:upb} it should be clear that $\alpha/b = \max\{d_{b}(e) \mid e \in P\}/b$ is a lower bound for $r(G,B^V)$ and can be computed in polynomial time.
	To establish the theorem, we argue that $(1+\tau)\alpha/b$ is an upper bound for $r(G,B^V)$.
	In particular, we prove that the cost configuration $\tilde{c}$ returned by {\sc VariablePresentBiasApprox} is motivating for all $\beta \in B^V$ if the reward is set to $(1+\tau)\alpha/b$.
	
	Using the same reasoning as in proof of Theorem~\ref{thm:upu}, we know that the agent's perceived cost at any node $v$ of $P$ is covered by a reward of $(1+\tau)\alpha/b$, i.e., $d_{\beta(v),\tilde{c}}(v) = \beta(v)((1+\tau)\alpha/b)$. 
	All that remains to be shown is that this reward is also sufficient if the agent is located on a node $v$ not on~$P$.
	For this purpose, let $v'$ be the last node of $P$ the agent visited before reaching $v$.
	Furthermore, let $(v,w)$ and $(v',w')$ denote the two edges the agent plans to take when located at $v$ and $v'$ respectively.
	According to the same argument we made in the proof of Theorem~\ref{thm:upu}, it holds true that a cheapest path from $w'$ to $t$ with respect to $\tilde{c}$ is more expensive than a cheapest path from $w$ to $t$, i.e., $d_{\tilde{c}}(w) \leq d_{\tilde{c}}(w')$.
	We also know that $\tau c(v,w) \leq \tilde{c}(v',w')$ and $\tilde{c}(v,w) = 0$ hold true by construction of $\tilde{c}$.
	Finally the definition of $\tau = \max B/\min B$ implies that $\beta(v) \geq \beta(v')/\tau$.
	Combining these inequalities yields
	\begin{align*}
		d_{\beta(v),\tilde{c}}(v)	&\leq d_{\beta(v),\tilde{c}}(v,w) = \beta(v)\Bigl(\frac{c(v,w)}{\beta(v)} + d(w)\Bigr) \leq \beta(v)\Bigl(\frac{c(v,w)}{\beta(v')/\tau} + d(w)\Bigr)\\
									&\leq \beta(v)\Bigl(\frac{\tilde{c}(v',w') + c(v'w')}{\beta(v')} + d(w')\Bigr) = \frac{\beta(v)}{\beta(v')} d_{\beta(v'),\tilde{c}}(v',w').
	\end{align*}
	
	Recall that the agent already traversed $(v',w')$.
	Therefore her perceived cost of $(v',w')$ is at most $d_{\beta(v'),\tilde{c}}(v',w') \leq \beta(v')((1+\tau)\alpha/b)$, implying that 
	\[d_{\beta(v),\tilde{c}}(v) \leq \frac{\beta(v)}{\beta(v')} \beta(v')\bigl((1+\tau)\alpha/b\bigr) = \beta(v)\bigl((1+\tau)\alpha/b\bigr).\]
	This completes the proof.
\end{proof}

Although {\sc VariablePresentBiasApprox} yields a good approximation for a moderately variable present bias, it does not provide a constant approximation bound like {\sc UncertainPresentBiasApprox}.
However, this is not necessarily a shortcoming of the algorithm, but a complexity consequence.
To prove hardness for constant factor approximations of VPB-OPT, we consider the VECTOR SCHEDULING (VS) problem.

\begin{definition}[VS-OPT]\label{def:vs}
	Given $m$ machines denoted by the sets $M_1,\ldots,M_m$ and $\ell$ jobs in the form of $d$-dimensional vectors $q_1,\ldots,q_{\ell} \in \mathbb{R}_{\geq 0}^d$, find the smallest makespan over all dimensions, i.e., minimize $\max\{\|\sum_{q \in M_i} q\|_\infty \mid 1 \leq i \leq m\}$ with respect to all partitions of the jobs into sets $M_1,\ldots,M_m$.
\end{definition}

As Chekuri and Khanna have shown, VS-OPT is unlikely to have an efficient constant factor approximation~\cite{CK}.

\begin{theorem}\label{thm:vshard}
	No polynomial time algorithm approximates VS-OPT within a constant factor ${\varrho>1}$, unless  $\rm{NP} = \rm{ZPP}$.
	This holds true even if the problem is restricted to $0$-$1$ vectors.
\end{theorem}

By reducing VS-OPT to VPB-OPT, we are able to prove that VPB-OPT is unlikely to have an efficient constant factor approximation as well.

\begin{theorem}\label{thm:mccvhard}
	No polynomial time algorithm approximates VPB-OPT within a constant factor $\varrho > 1$, unless $\rm{NP} = \rm{ZPP}$.
\end{theorem}

\begin{proof}
	Our proof is based on the following reduction from VS-OPT.
	Let $\mathcal{I}$ be an arbitrary instance of VS-OPT with a set of $m$ machines and $\ell \geq 2$ jobs $q_1,\ldots,q_{\ell} \in \{0,1\}^d$.
	Our goal is to construct an instance $\mathcal{J}$ of VPB-OPT with a size polynomial in that of $\mathcal{I}$ such that the following two claims are satisfied:
	(a) If $\mathcal{I}$ has a schedule with a makespan of $\kappa$, then $\mathcal{J}$ has a cost configuration that is motivating for a reward of $r = \kappa\ell + \ell + 1$.
	(b) If $\mathcal{J}$ has a cost configuration that is motivating for a reward of $r$, then $\mathcal{I}$ has a schedule with a makespan of at most $\kappa = 2r/\ell$.
	Consequently, if we have a polynomial time algorithm for VPB-OPT with a constant approximation ratio $\varrho$, we can apply this algorithm to $\mathcal{J}$ and recover a $2\varrho + \mathcal{O}(1)$ approximate solution for $\mathcal{I}$ in polynomial time.
	According to Theorem \ref{thm:vshard} this is not possible unless $\rm{NP} = \rm{ZPP}$.
	
	As our first step we need to construct the VPB-OPT instance $\mathcal{J}$.
	In particular, we need to specify a present bias set $B$ and a task graph $G$.
	We begin with $B$.
	From Theorem~\ref{thm:vaprox} we know that VPB-OPT has a $(1+\tau)$-approximation algorithm.
	Consequently, we should choose $B$ in such a way that its range depends on the size of $\mathcal{I}$.
	Furthermore, Proposition~\ref{prop:cvs} implies that $B$ should not include the value $1$.
	For the sake of simplicity we set $B = \{1/2,1/\ell^2\}$.
	
	\begin{figure}[t]
		\center
		\begin{tikzpicture}[scale=0.875, nst/.style={draw,circle,fill=black,minimum size=4pt,inner sep=0pt]}, est/.style={draw,>=latex,->}]
			
			\node[nst] (v2) at (2.5,0) [label=left:\small{$v_{i,j,1}$}] {};
			\node[nst] (v3) at (5,0) {};
			\node[nst] (v4) at (7.5,0) {};
			\node[nst] (v5) at (10,0) {};
			\node[nst] (v6) at (12.5,0) [label=right:\small{$w_{i,j,1}$}] {};
			\node[nst] (v8) at (2.5,2) [label=left:\small{$v_{i,j,2}$}] {};
			\node[nst] (v9) at (5,2) {};
			\node[nst] (v10) at (7.5,2) {};
			\node[nst] (v11) at (10,2) {};
			\node[nst] (v12) at (12.5,2) [label=right:\small{$w_{i,j,2}$}] {};
			\node[nst] (v14) at (2.5,4) [label=left:\small{$v_{i,j,3}$}] {};
			\node[nst] (v15) at (5,4) {};
			\node[nst] (v16) at (7.5,4) {};
			\node[nst] (v17) at (10,4) {};
			\node[nst] (v18) at (12.5,4) [label=right:\small{$w_{i,j,3}$}] {};
			\node[nst] (v20) at (2.5,6) [label=left:\small{$v_{i,j,\ell}$}] {};
			\node[nst] (v21) at (5,6) {};
			\node[nst] (v22) at (7.5,6) {};
			\node[nst] (v23) at (10,6) {};
			\node[nst] (v24) at (12.5,6) [label=right:\small{$w_{i,j,\ell}$}] {};
			\node[nst] (v25) at (2.5,8) [label=above:\small{$t$}] {};
			\node[nst] (w1) at (2.5,1) {};
			\node[nst] (w2) at (2.5,3) {};
			\node[nst] (w3) at (2.5,7) {};

			\node at (8.75,0) {$\dots$};
			\node at (8.75,1) {$\dots$};		
			\node at (8.75,2) {$\dots$};
			\node at (8.75,3) {$\dots$};
			\node at (8.75,4) {$\dots$};
			\node at (8.75,6) {$\dots$};
			\node at (8.75,7) {$\dots$};
			\node at (2.5,5) {$\vdots$};
			\node at (5,5) {$\vdots$};
			\node at (7.5,5) {$\vdots$};
			\node at (10,5) {$\vdots$};
			\node at (12.5,5) {$\vdots$};
			
			\path (v2) edge[est] node [below] {\small{$1/\ell^2$}} (v3)
			(v3) edge[est] node [below] {\small{$1/\ell^2$}} (v4)
			(v4) edge (8.25,0)
			(9.25,0) edge[est] (v5)
			(v5) edge[est] node [below] {\small{$1/\ell^2$}} (v6)
			(v2) edge[est] node [left] {\small{$0$}} (w1)
			(v3) edge node [left] {\small{$\ell$}} (5,1)
			(v4) edge node [left] {\small{$\ell$}} (7.5,1)
			(v5) edge node [left] {\small{$\ell$}} (10,1)
			(v6) edge node [left] {\small{$\ell$}} (12.5,1)
			(12.5,1) edge (9.25,1)
			(8.25,1) edge (3.5,1)
			(3.5,1) edge[est] (v8)
			(w1) edge[est] node [left] {\small{$1$}} (v8)
			(v8) edge[est] node [below] {\small{$1/\ell^2$}} (v9)
			(v9) edge[est] node [below] {\small{$1/\ell^2$}} (v10)
			(v10) edge (8.25,2)
			(9.25,2) edge[est] (v11)
			(v11) edge[est] node [below] {\small{$1/\ell^2$}} (v12)
			(v8) edge[est] node [left] {\small{$0$}} (w2)
			(v9) edge node [left] {\small{$\ell$}} (5,3)
			(v10) edge node [left] {\small{$\ell$}} (7.5,3)
			(v11) edge node [left] {\small{$\ell$}} (10,3)
			(v12) edge node [left] {\small{$\ell$}} (12.5,3)
			(12.5,3) edge (9.25,3)
			(8.25,3) edge (3.5,3)
			(3.5,3) edge[est] (v14)
			(w2) edge[est] node [left] {\small{$1$}} (v14)
			(v14) edge[est] node [below] {\small{$1/\ell^2$}} (v15)
			(v15) edge[est] node [below] {\small{$1/\ell^2$}} (v16)
			(v16) edge (8.25,4)
			(9.25,4) edge[est] (v17)
			(v17) edge[est] node [below] {\small{$1/\ell^2$}} (v18)
			(v14) edge (2.5,4.5)
			(v15) edge (5,4.5)
			(v16) edge (7.5,4.5)
			(v17) edge (10,4.5)
			(v18) edge (12.5,4.5)
			(2.5,5.5) edge[est] (v20)
			(3,5.5) edge[est] (v20)
			(v20) edge[est] node [below] {\small{$1/\ell^2$}} (v21)
			(v21) edge[est] node [below] {\small{$1/\ell^2$}} (v22)
			(v22) edge (8.25,6)
			(9.25,6) edge[est] (v23)
			(v23) edge[est] node [below] {\small{$1/\ell^2$}} (v24)
			(v20) edge[est] node [left] {\small{$0$}} (w3)
			(v21) edge node [left] {\small{$\ell$}} (5,7)
			(v22) edge node [left] {\small{$\ell$}} (7.5,7)
			(v23) edge node [left] {\small{$\ell$}} (10,7)
			(v24) edge node [left] {\small{$\ell$}} (12.5,7)
			(12.5,7) edge (9.25,7)
			(8.25,7) edge (3.5,7)
			(3.5,7) edge[est] (v25)
			(w3) edge[est] node [left] {\small{$1$}} (v25);
		\end{tikzpicture}
		\caption{A column $H_{i,j}$ of the task graph $G$}\label{fig:col}
	\end{figure}
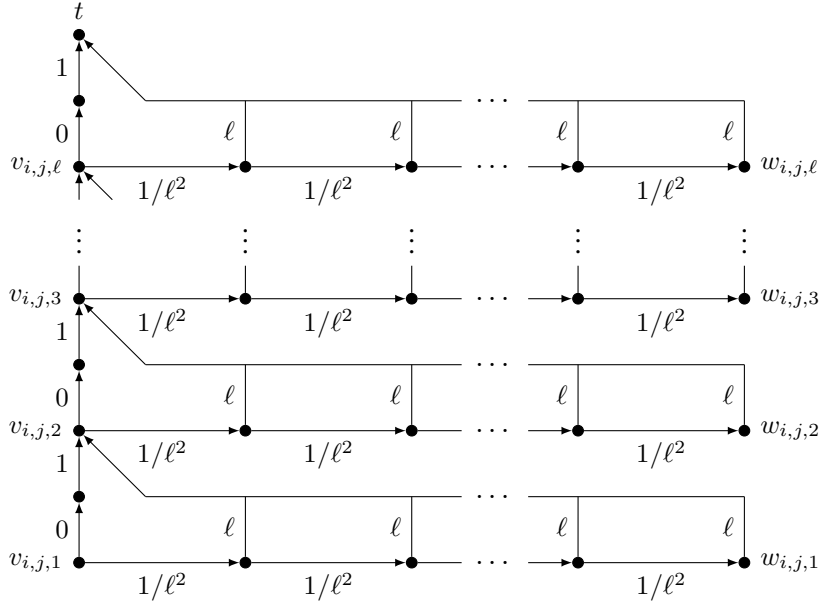		
	
	We continue with the construction of $G$.
	For this purpose we introduce a {\em column} $H_{i,j}$ for each machine $i$ and dimension $j$.
	The structure of these columns, which is illustrated in Figure~\ref{fig:col}, is identical for all machines and dimensions.
	More precisely, each column $H_{i,j}$ consists of $\ell$ {\em levels} and each level $k$, with $1 \leq k \leq \ell$, consists of a directed path starting at a node $v_{i,j,k}$ and ending at $w_{i,j,k}$.
	In between there are $\ell^4$ edges of cost $1/\ell^2$ resulting in a total path cost of $\ell^2$.
	Furthermore, each node of the path has a {\em shortcut} to the next level $k+1$, or to $t$ in the special case of $k = \ell$.
	For a convenient notation let $t = v_{i,j,\ell+1}$.
	We distinguish between two types of shortcuts.
	Shortcuts of the {\em first type} connect $v_{i,j,k}$ with $v_{i,j,k+1}$ via two edges of cost $0$ and $1$ respectively.
	Shortcuts of the {\em second type} connect all remaining nodes on level $k$ with $v_{i,j,k+1}$ via a single edge of cost $\ell$.
	For the sake of a concise representation Figure~\ref{fig:col} merges some shortcuts.
	
	To connect the individual columns, we construct a path $P_{i,k}$ for each machine $i$ and job $q_k$.
	The idea of $P_{i,k}$ is to cross the $k$-th level of all columns $H_{i,j}$ for which $q_k$ has a cost of $1$ in dimension $j$, i.e., $(q_k)_j = 1$.
	As start and endpoint of $P_{i,k}$ we introduce additional nodes $u_{i-1}$ and $u_{i}$ and set $s = u_0$ and $t = u_\ell$.
	A more formal construction of $P_{i,k}$ is as follows:
	Let $j$ be an arbitrary dimension for which $(q_k)_j = 1$.
	Without loss of generality we may assume that at least one such dimension exists.
	Otherwise, $q_k$ could be assigned to any machine without affecting the makespan and would therefore be irrelevant to the schedule.
	If $j$ is the dimension of lowest index satisfying $(q_k)_j = 1$, we draw an edge of cost $1/\ell^2$ from $u_{k-1}$ to $v_{i,j,k}$.
	Similarly, we draw an edge of cost $1/\ell^2$ from $w_{i,j,k}$ to $u_{k}$ if $j$ is the dimension of highest index for which $(q_k)_j = 1$.
	For all intermediate dimensions $j$ satisfying $(q_k)_j = 1$, we draw an edge of cost $1/\ell^2$ from $w_{i,j,k}$ to a distinct intermediate node $u_{i,j,k}$ and another edge of the same cost from $u_{i,j,k}$ to $v_{i,j',k}$ with $j'$ being the dimension of next higher index satisfying $(q_k)_{j'} = 1$.
	Figure~\ref{fig:red} illustrates this construction for a small sample instance of $\mathcal{I}$.
	Note that the nodes $u_1$ and $u_2$ appear twice in the figure, once on the left and once on the right, to keep the drawing simple.
	
	To complete our construction, we introduce a {\em third type} of shortcuts connecting all nodes $u_k \neq t$ and $u_{i,j,k}$ to $t$ via a single edge of cost $\ell$.
	For the sake of a clear representation these shortcuts are not depicted in Figure~\ref{fig:red}.
	Note that the resulting task graph $G$ is acyclic and can be constructed in polynomial time with respect to $\mathcal{I}$.
	It should also be mentioned that some columns of $G$ might contain nodes which are not reachable from $s$.
	However, as argued in Section~\ref{sec:model} we may ignore these nodes as they do not affect the agent's behavior and can be removed from $G$ in polynomial time.
	
	We proceed with the proof of statement (a), i.e., given a partition of jobs $M_1,\ldots,M_m$ for $\mathcal{I}$ resulting in a makespan $\kappa$, we show how to construct a cost configuration $\tilde{c}$ such that $\mathcal{J}$ is motivating for a reward of $\kappa\ell + \ell + 1$.
	The construction of $\tilde{c}$ is simple.
	For each job $q_k$ it is sufficient to assign an extra cost of $\ell$ to the initial edge of all shortcuts starting at a node $v_{i,j,k}$ for which $i$ is the machine $q_k$ is scheduled on, i.e., $q_k \in M_i$, and $j$ is a dimension in which $q_k$ has a cost of $1$, i.e., $(q_k)_j = 1$.
	Furthermore, an extra cost of $\kappa\ell + \ell + 2$ needs to be assigned to the initial edge of all paths $P_{i',k}$ associated with the machines $i'$ job $q_k$ is not scheduled on, i.e., $q_k \notin M_{i'}$.
	
	To show that $\tilde{c}$ is motivating for a reward of $\kappa\ell + \ell + 1$, we argue that an agent with a present bias configuration $\beta \in B^V$ successfully constructs a path from $s$ to $t$ that only consists of paths $P_{i,k}$ for which $q_k \in M_{i}$ and possibly a shortcut of the third type.
	For this purpose, assume the agent is located at some node $v \neq t$ on $P_{i,k}$.
	A case distinction with respect to the type of $v$ shows that the agent either stays on $P_{i,k}$ or takes a shortcut of the third type.
	As the paths $P_{i,k}$ are connected at their terminal nodes, this proves the claim.
	
	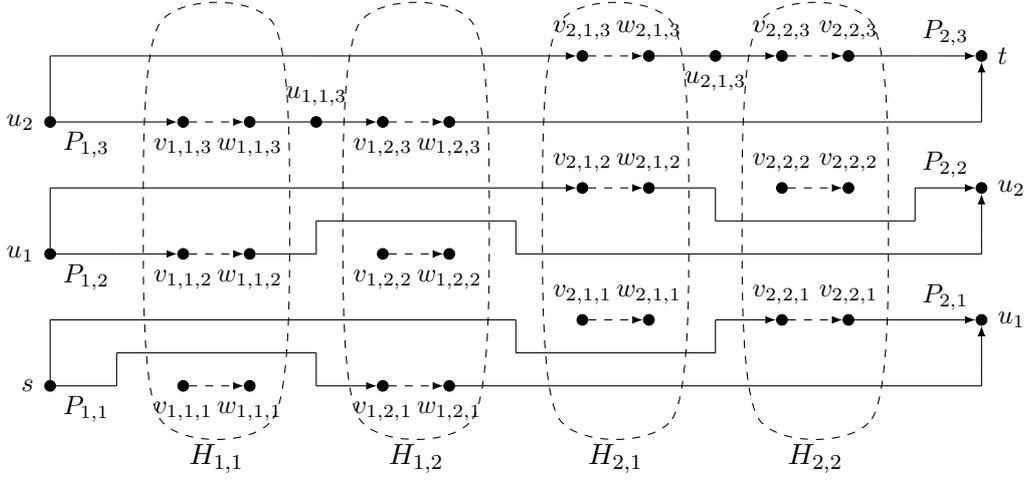
\begin{figure}[t]
		\center
		\begin{tikzpicture}[scale=0.875, nst/.style={draw,circle,fill=black,minimum size=4pt,inner sep=0pt]}, est/.style={draw,>=latex,->}]
			
			\node[nst] (v1) at (0,0) [label=left:\small{$s$}] {};
			\node[nst] (v2) at (2,0) [label=below:\small{$v_{1,1,1}$}] {};
			\node[nst] (v3) at (3,0) [label=below:\small{$w_{1,1,1}$}] {};
			\node[nst] (v4) at (5,0) [label=below:\small{$v_{1,2,1}$}] {};
			\node[nst] (v5) at (6,0) [label=below:\small{$w_{1,2,1}$}] {};
			\node[nst] (v6) at (8,1) [label=above:\small{$v_{2,1,1}$}] {};
			\node[nst] (v7) at (9,1) [label=above:\small{$w_{2,1,1}$}] {};
			\node[nst] (v8) at (11,1) [label=above:\small{$v_{2,2,1}$}] {};
			\node[nst] (v9) at (12,1) [label=above:\small{$v_{2,2,1}$}] {};
			\node[nst] (v10) at (2,2) [label=below:\small{$v_{1,1,2}$}] {};
			\node[nst] (v11) at (3,2) [label=below:\small{$w_{1,1,2}$}] {};
			\node[nst] (v12) at (5,2) [label=below:\small{$v_{1,2,2}$}] {};
			\node[nst] (v13) at (6,2) [label=below:\small{$w_{1,2,2}$}] {};
			\node[nst] (v14) at (8,3) [label=above:\small{$v_{2,1,2}$}] {};
			\node[nst] (v15) at (9,3) [label=above:\small{$w_{2,1,2}$}] {};
			\node[nst] (v16) at (11,3) [label=above:\small{$v_{2,2,2}$}] {};
			\node[nst] (v17) at (12,3) [label=above:\small{$v_{2,2,2}$}] {};
			\node[nst] (v18) at (2,4) [label=below:\small{$v_{1,1,3}$}] {};
			\node[nst] (v19) at (3,4) [label=below:\small{$w_{1,1,3}$}] {};
			\node[nst] (v20) at (5,4) [label=below:\small{$v_{1,2,3}$}] {};
			\node[nst] (v21) at (6,4) [label=below:\small{$w_{1,2,3}$}] {};
			\node[nst] (v22) at (8,5) [label=above:\small{$v_{2,1,3}$}] {};
			\node[nst] (v23) at (9,5) [label=above:\small{$w_{2,1,3}$}] {};
			\node[nst] (v24) at (11,5) [label=above:\small{$v_{2,2,3}$}] {};
			\node[nst] (v25) at (12,5) [label=above:\small{$v_{2,2,3}$}] {};
			\node[nst] (v26) at (14,5) [label=right:\small{$t$}] {};
			\node[nst] (v27) at (14,1) [label=right:\small{$u_1$}] {};
			\node[nst] (v28) at (0,2) [label=left:\small{$u_1$}] {};
			\node[nst] (v29) at (14,3) [label=right:\small{$u_2$}] {};
			\node[nst] (v30) at (0,4) [label=left:\small{$u_2$}] {};
			\node[nst] (v31) at (4,4) [label=above:\small{$u_{1,1,3}$}] {};
			\node[nst] (v32) at (10,5) [label=below:\small{$u_{2,1,3}$}] {};

			\draw[dashed] plot [smooth cycle] coordinates {(1.65,-0.5) (3.35,-0.5) (3.6,2.5) (3.35,5.5) (1.65,5.5) (1.4,2.5)};
			\draw[dashed] plot [smooth cycle] coordinates {(4.65,-0.5) (6.35,-0.5) (6.6,2.5) (6.35,5.5) (4.65,5.5) (4.4,2.5)};
			\draw[dashed] plot [smooth cycle] coordinates {(7.65,-0.5) (9.35,-0.5) (9.6,2.5) (9.35,5.5) (7.65,5.5) (7.4,2.5)};
			\draw[dashed] plot [smooth cycle] coordinates {(10.65,-0.5) (12.35,-0.5) (12.6,2.5) (12.35,5.5) (10.65,5.5) (10.4,2.5)};
			\node at (2.5,-1.1) {$H_{1,1}$};
			\node at (5.5,-1.1) {$H_{1,2}$};
			\node at (8.5,-1.1) {$H_{2,1}$};
			\node at (11.5,-1.1) {$H_{2,2}$};
			
			\path (v1) edge node [below] {\small{$P_{1,1}$}} (1,0)
			(1,0) edge (1,0.5)
			(1,0.5) edge (4,0.5)
			(4,0.5) edge (4,0)
			(4,0) edge[est] (v4)
			(v5) edge (14,0)
			(14,0) edge[est] (v27)
			(v1) edge (0,1)
			(0,1) edge (7,1)
			(7,1) edge (7,0.5)
			(7,0.5) edge (10,0.5)
			(10,0.5) edge (10,1)
			(10,1) edge[est] (v8)
			(v9) edge (13,1) 
			(13,1) edge[est] node [above] {\small{$P_{2,1}$}} (v27)
			(v28) edge node [below] {\small{$P_{1,2}$}} (1,2)
			(1,2) edge[est] (v10)
			(v11) edge (4,2)
			(4,2) edge (4,2.5)
			(4,2.5) edge (7,2.5)
			(7,2.5) edge (7,2)
			(7,2) edge (14,2)
			(14,2) edge[est] (v29)
			(v28) edge (0,3)
			(0,3) edge[est] (v14)
			(v15) edge (10,3)
			(10,3) edge (10,2.5)
			(10,2.5) edge (13,2.5)
			(13,2.5) edge (13,3)
			(13,3) edge[est] node [above] {\small{$P_{2,2}$}} (v29)
			(v30) edge node [below] {\small{$P_{1,3}$}} (1,4) 
			(1,4) edge[est] (v18)
			(v19) edge[est] (v20)
			(v21) edge (14,4)
			(14,4) edge[est] (v26)
			(v30) edge (0,5)
			(0,5) edge[est] (v22)
			(v23) edge[est] (v24)
			(v25) edge (13,5)
			(13,5) edge[est] node [above] {\small{$P_{2,3}$}} (v26)
			(v2) edge[est,dashed] (v3)
			(v4) edge[est,dashed] (v5)
			(v6) edge[est,dashed] (v7)
			(v8) edge[est,dashed] (v9)
			(v10) edge[est,dashed] (v11)
			(v12) edge[est,dashed] (v13)
			(v14) edge[est,dashed] (v15)
			(v16) edge[est,dashed] (v17)
			(v18) edge[est,dashed] (v19)
			(v20) edge[est,dashed] (v21)
			(v22) edge[est,dashed] (v23)
			(v24) edge[est,dashed] (v25);
		\end{tikzpicture}
		\caption{Reduction from a VS instance with $2$ machines and jobs $q_1 = (0,1)$, $q_2 = (1,0)$, $q_3 = (1,1)$}\label{fig:red}
	\end{figure}
	
	We begin with the case $v = u_{k-1}$.
	When located at such a node $v$ the agent has three choices.
	First, she might plan to enter a path $P_{i',k}$ for which $i' \neq i$.
	However, $\tilde{c}$ assigns an extra cost of $\kappa\ell + \ell + 2$ to the initial edge of $P_{i',k}$.
	Therefore the agent is not motivated to take $P_{i',k}$ no matter what the exact value of $\beta(v)$ is.
	Secondly, the agent might plan to enter $P_{i,k}$ from $v$.
	If she plans to follow the shortcuts of the first kind all the way to $t$ after traversing the first edge of $P_{i,k}$, she faces an immediate edge of cost $1/\ell^2$ and at most $\ell$ future edges of cost $1$, one per subsequent level.
	Furthermore, the number of future edges with an extra cost of $\ell$ is bounded by the makespan $\kappa$.
	Considering that $1/\ell^2 \leq \beta(v)$, we conclude that the agent's perceived cost for entering $P_{i,k}$ is at most $1/\ell^2 + \beta(v) (\ell + \kappa\ell) \leq \beta(v) + \beta(v)(\ell + \kappa\ell) = \beta(v)(\kappa\ell + \ell + 1)$.
	As this matches her perceived reward, entering $P_{i,k}$ must be motivating.
	Thirdly, the agent might plan to take the shortcut of the third type from $v$.
	If she does so, she immediately reaches $t$.
	If not, the only motivating option that remains is $P_{i,k}$.
	Either way our claim holds true.
	
	We continue with the case $v = u_{i,j,k}$, which is very similar to the previous one.
	Again the agent may either stay on $P_{i,k}$ or take the shortcut of the third type.
	Analyzing the perceived cost of these two options is identical to the previous case.
	The only difference is that the agent is missing the first option, which is to enter a path $P_{i',k}$ for which $i' \neq i$.
	However, this does not affect her motivation to stay on $P_{i,k}$ nor does it affect the possibility of her taking the shortcut.	
	
	Next we consider the case that $v$ is a node on the $k$-th level of some column $H_{i,j}$ different from~$w_{i,j,k}$.
	When located at $v$ the agent has two choices.
	Either she takes the immediate shortcut incident to $v$ or she traverses an edge of $P_{i,k}$.
	Her perceived cost of the first option is $\ell + \beta(v) d_{\tilde{c}}(v_{i,j,k+1})$ while her perceived cost of the second option is at most $1/\ell^2 + \beta(v) (\ell + d_{\tilde{c}}(v_{i,j,k+1}))$ if she plans to take the shortcut of the second type immediately after traversing the edge of $P_{i,k}$.
	Consequently, the agent prefers to stay on $P_{i,k}$ whenever $\ell > 1/\ell^2 + \beta(v)\ell$ or $1 - 1/\ell^3 > \beta(v)$ if we rearrange the inequality.
	Considering that $\ell \geq 2$ and $\beta(v) \leq 1/2$ hold true, this inequality is certainly satisfied.
	It remains to show that staying on $P_{i,k}$ is also a motivating option.
	However, following the same line of argument we have used for the perceived cost of $P_{i,k}$ in the case $v = u_{i,j,k}$, this should be easy to see.
	
	Finally, we consider the remaining case $v = w_{i,j,k}$.
	Similar to the previous case the agent has two options.
	Either she takes the immediate shortcut incident to $v$ or she traverses $P_{i,k}$ and ends up at $u_{i,j,k}$.
	Her perceived cost of the first option is $\ell + \beta(v) d_{\tilde{c}}(v_{i,j,k+1}) \geq \ell$.
	In contrast, if she moves to $u_{i,j,k}$ and then takes the shortcut of the third kind immediately after, her perceived cost is ${1/\ell^2 + \beta(v) \ell}$.
	As a result she prefers the first option whenever $\ell > 1/\ell^2 + \beta(v)\ell$.
	But from the previous case we already know that this inequality is always satisfied.
	Furthermore, the fact that $1/\ell^2 \leq \beta(v)$ implies that the agent's perceived cost for moving to $u_{i,j,k}$ is at most $1/\ell^2 + \beta(v)\ell \leq \beta(v) + \beta(v)\ell \leq \beta(v)(\kappa\ell + \ell + 1)$.
	As this matches her perceived reward, we know that moving to $u_{i,j,k}$ is motivating.
	
	We now come to the proof of (b).
	For this purpose assume $\mathcal{J}$ has a cost configuration $\tilde{c}$ that is motivating for a reward of $r$.
	Our goal is to schedule the jobs of $\mathcal{I}$ with a makespan of at most~$2r/\ell$.
	Since any schedule has a makespan less or equal to $\ell$, we focus on the case $r < \ell^2/2$.
	As a result, the agent's perceived reward becomes $r/\ell^2 < 1/2$ whenever her present bias takes the value~$1/\ell^2$.
	However, this implies that the agent must not take shortcuts of the first kind.
	The reason is that her perceived cost at the intermediate node of the shortcut is at least $1$ and therefore she might lose motivation.
	Furthermore, the agent is not motivated to take shortcuts of the second or third type whenever her present bias is $1/\ell^2$.
	Consequently, if we fix the agent's present bias to $\beta(v) = 1/\ell^2$ at all nodes $v$ of $G$, she must construct a path $P$ from $s$ to $t$ that does not contain shortcuts of any type.
	This means that $P$ can be divided into a sequence of paths $P_{i,k}$.
	By assigning job $q_k$ to machine $i$ if $P_{i,k}$ is contained in $P$ we obtain a feasible partition $M_1,\ldots,M_m$.	
	
	Next we argue that the makespan of $M_1,\ldots,M_m$ is at most $2r/\ell$.
	For this purpose we first show that $\tilde{c}$ must assign an extra cost of at least $\ell/2 - 1$ to all shortcuts starting at a node $v_{i,j,k}$ for which $q_{k}$ is scheduled on $i$, i.e., $q_{k} \in M_i$.
	For this purpose, assume the agent is located at $v_{i,j,k}$.
	Note that this scenario is indeed possible if $\beta(v) = 1/\ell^2$ for all nodes $v$ of $P$ that come before~$v_{i,j,k}$.
	Furthermore, assume that $\beta(v_{i,j,k}) = 1/2$.
	As we have argued in the previous paragraph, the agent must not take the shortcut at $v_{i,j,k}$, but stay on $P_{i,k}$.
	Let $P'$ denote her planned path.
	We distinguish between two possible scenarios for $P'$.
	First, $P'$ might follow the path $P_{i,k}$ until the agent reaches another column or the node $u_k$.
	In this case, the first $\ell^4 + 1$ edges of $P'$ are all of cost~$1/\ell^2$.
	As a result, the perceived cost of $P'$ is at least ${1/\ell^2 + \beta(v_{i,j,k})\ell^2 = 1/\ell^2 + \ell^2/2}$.
	Considering that $r < \ell^2/2$, this cannot be motivating.
	Secondly, $P'$ might contain a shortcut of the second type to the next level of $H_{i,j}$.
	Even if we neglect potential extra cost $\tilde{c}$ may assign to the current level of $H_{i,j}$ and furthermore assume that the agent takes the very next shortcut, her perceived cost of $P'$ is at least $1/\ell^2 + \beta(v_{i,j,k})(\ell + d_{\tilde{c}}(v_{i,j,k+1}))$ in this case.
	In contrast, the agent's perceived cost for taking the immediate shortcut at $v_{i,j,k}$ is only $\beta(v_{i,j,k})(1 + d_{\tilde{c}}(v_{i,j,k+1}))$ if we neglect potential extra cost.
	Since the agent must not enter the shortcut at $v_{i,j,k}$, we conclude that the cost configuration $\tilde{c}$ assigns an extra cost greater than $1/\ell^2 + \beta(v_{i,j,k})(\ell - 1) > \ell/2 - 1$ to the shortcut.
	
	To see that the makespan of our schedule is at most $2r/\ell$, consider the workload $\kappa$ on an arbitrary machine $i$ in an arbitrary dimension $j$.
	Assuming that $\kappa > 0$, let $q_k$ be the job of lowest index scheduled on $i$ such that $(q_k)_i = 1$.
	Furthermore, assume that the agent is located at $v_{i,j,k}$ and that her current present bias is $\beta(v_{i,j,k}) = 1/\ell^2$.
	As argued in the previous paragraph, this assumption is justified.
	We continue with a case distinction on the path $P'$ the agent plans when located at $v_{i,j,k}$.
	In general, $P'$ can have one of the following two forms:
	Either it exits column $H_{i,k}$ to the right or it is completely contained in $H_{i,k}$ and climbs to the top level via shortcuts of the first and or second type.
	If $P'$ exits to the right, it must contain at least $\ell^4 + 1$ edges costing $1/\ell^2$ each.
	Clearly, the agent's perceived cost of such a $P'$ is at least $1/\ell^2 + \beta(v_{i,j,k})(\ell^4/\ell^2) > 1$ and cannot be motivating for a perceived reward $\beta(v_{i,j,k})r = r/\ell^2 < 1/2$.
	However, this means that some path $P'$ of the second form must be motivating.
	In this case, $P'$ consists of at least $\kappa$ shortcuts that according to the result from the previous paragraph all have a cost of at least $\ell/2$ with respect to $\tilde{c}$.
	As a result, the perceived cost of $P'$ is at least $\beta(v_{i,j,k})\kappa\ell/2 = \kappa/(2\ell)$.
	To make sure that this cost does not exceed $\beta(v_{i,j,k})r = r/\ell^2$, the total load $\kappa$ of machine $i$ in dimension $j$ can be at most $2r/\ell$.
\end{proof}

\subsection{Occasionally Unbiased Agents}

Although VPB is hard to solve in general, a curious special case consisting of all present bias sets $B$ for which $1 \in B$ is not.
Note that agents whose present bias varies within such a $B$ becomes temporarily unbiased whenever $1$ is drawn.
For this reason we call these agents {\em occasionally unbiased}.
A behavioral pattern unique to occasionally unbiased agents is that they may start to walk along a cheapest path at any point in time whenever their present bias becomes $1$.
As a result we can reduce VPB to a decision problem we call CRITICAL NODE SET (CNS) for occasionally unbiased agents.

\begin{definition}[CNS]
	Given a task graph $G$, present bias set $B$ and reward $r$, decide the existence of a critical node set $W$. 
\end{definition}

We consider a node set $W$ {\em critical} if the following properties hold:
(a) $s \in W$.
(b) Each node $v \in W$ has a path $P$ to $t$ that only uses nodes of $W$.
(c) All edges $e$ of $P$ satisfy $d_{b}(e) \leq br$ with $b = \min B$.
As it turns out, such a $W$ contains exactly those nodes an occasionally unbiased agent may visit with respect to a motivating cost configuration.
This allows us to reduce VPB to CNS.

\begin{proposition}\label{prop:cvs}
	Assuming that $1 \in B$, then VPB has a solution if and only if CNS has one.
\end{proposition}

\begin{proof}
	($\Rightarrow$) To prove the first implication, assume a critical node set $W$ exists.
	Our goal is to construct a cost configuration $\tilde{c}$ that is motivating for all $\beta \in B^V$.
	For this purpose, we assign an extra cost of $\tilde{c}(v,w) = r+1$ to all edges $(v,w)$ that leave $W$, i.e., all edges for which $v \in W$ but $w \notin W$, and set $\tilde{c}(e) = 0$ for the remaining edges.
	The resulting $\tilde{c}$ is motivating for two reasons:
	(a)~Whenever the agent is inside $W$, she is not motivated to leave $W$.
	(b)~Each node $v \in W \setminus \{t\}$ has at least one successor $w \in W$ the agent is motivated to visit next.
	Together with the fact that $s \in W$ we conclude that the agent moves through $W$ until she eventually reaches $t$.
	
	In the following we show (a) and (b).
	Since each edge $(v,w)$ that leaves $W$ has an extra cost of $r + 1$, the perceived cost of $(v,w)$ trivially exceeds the agent's perceived reward of $\beta(v) r$.
	Therefore $(v,w)$ cannot be motivating, which proves (a).
	We proceed with (b). 
	For each node $v \in W \setminus \{t\}$ we need to show the existence of a direct successor $w \in W$ such that $(v,w)$ is motivating for all $\beta(v) \in B$.
	To prove this, remember that $W$ is a critical node set.
	Consequently, there must exist a path $P = v,w,\ldots,t$ that consists exclusively of nodes from $W$ such that $c(v,w) + b \sum_{e \in P'} c(e) \leq br$, with $P'$ being the path obtained by removing the initial edge of $P$.
	Because $P$ is contained in $W$, no edge of $P$ charges extra cost.
	By definition of $b$ we also know that $\beta(v)/b \geq 1$ and we can bound the perceived cost of $(v,w)$ from above by
	\[d_{\beta(v),\tilde{c}}(v,w) \leq c(v,w) +  \beta(v) \sum_{e \in P'} c(e) \leq \frac{\beta(v)}{b} \Bigl(c(v,w) + b \sum_{e \in P'} c(e)\Bigr) \leq \beta r.\]
	As this bound matches the perceived reward, $(v,w)$ must be motivating.
	
	($\Leftarrow$) To prove the reverse, assume that $G$ has a cost configuration $\tilde{c}$ such that $G_{\tilde{c}}$ is motivating for all present bias configurations $\beta \in B^V$.
	Moreover, let $W'$ be the set of all nodes the agent might visit on her path from $s$ to $t$.
	Our goal is to argue that $W'$ is a critical node set, i.e., $W'$ meets the following requirements:
	(a)~$s \in W'$.
	(b)~Each node $v \in W'$ has a path $P$ to $t$ that only uses nodes of~$W'$.
	(c)~All edges $e$ of $P$ satisfy $d_{b}(e) \leq br$ with $b = \min B$.
	By definition of $W'$ (a) is trivially satisfied.
	To prove (b) and (c), let $v$ be some node of $W'$ and choose $w$ as the immediate successor of $v$ the agent visits for a present bias of $\beta(v) = b$.
	Furthermore, let $P'$ be a cheapest path from $w$ to $t$ with respect to $\tilde{c}$.
	By adding $(v,w)$ to the initial node of $P'$, we obtain a path $P$ satisfying (b) and~(c).

\begin{algorithm}[t]\label{alg:dcns}
	\caption{\sc DecideCriticalNodeSet}
	\SetKw{KwRevTopOrd}{in reverse topological order}
	\SetKw{KwElseRet}{else return}
	\SetKw{KwElse}{else}
    $\delta(t) \gets 0$\;
	\ForEach{$v \in V \setminus \{t\}$ \KwRevTopOrd}{
		$U \gets \{w \mid (v,w) \in E \text{ and } c(v,w) + \beta \delta(w) \leq b\}$\;
		\lIf{$U = \emptyset$}{$\delta(v) \gets \infty$; \KwElse $\delta(v) \gets \min\{c(v,w) + \delta(w) \mid w \in U\}$}
	}
	\lIf{$\delta(s) < \infty$}{\Return ``yes'' \KwElseRet ``no''}
\end{algorithm}
	
	We first show (c).
	Recall that the agent is motivated to traverse $(v,w)$ for a present bias of $\beta(v) = b$.
	The fact that $P'$ is a cheapest path from $w$ to $t$ with respect to $\tilde{c}$ immediately implies that the perceived cost of $P$ is at most
	\[c(v,w) + b \sum_{e \in P} c(e) \leq c(v,w) + \tilde{c}(v,w) + b \sum_{e \in P}\bigl(c(e)\ + \tilde{c}(e)\bigr) = d_{b,\tilde{c}}(v,w) \leq br.\]
	We continue with (b).
	Assuming the agent is located at $v$, consider a present bias configuration $\beta$ that assigns a value of $\beta(v) = b$ to $v$ and $\beta(v') = 1$ to all nodes $v' \neq v$ located on paths from $v$ to~$t$.
	Because $1 \in B$, such a present bias configuration is possible.
	By choice of $(v,w)$, we also know that the agent potentially traverses $w$ and at that point follows a cheapest path to $t$ with respect to $\tilde{c}$.
	Since $P'$ is such a path, we conclude that the agent may visit any node of $P$ and therefore all nodes of $P$ are certainly contained in $W'$.
\end{proof}

All that remains to show is that CNS is decidable in polynomial time.
A straight forward approach to this simple algorithmic problem is {\sc DecideCriticalNodeSet}, see Algorithm~\ref{alg:dcns}.
The main idea of this algorithm is simple.
Going through each node $v$ of $G$ in reverse topological order, {\sc DecideCriticalNodeSet} computes the values $\delta(v)$ denoting the cost of a cheapest path from $v$ to $t$ with respect to the $\delta$ values computed so far.
If the perceived cost at $v$ with respect to $\delta$ is too expensive, then $v$ cannot be contained in any critical node set and the algorithm sets $\delta(v) = \infty$.
Consequently, only nodes with $\delta(v) < \infty$ can be part of a critical node set.
In particular, we know that a critical node set exists if and only if $\delta(s) < \infty$.
Clearly, {\sc DecideCriticalNodeSet} can be executed in polynomial time which proves that CNS, can be decided in polynomial time for occasionally unbiased agents.
As a corollary of Proposition~\ref{prop:cvs}, so can VPB.

\begin{corollary}
	If $1 \in B$, then VPB can be solved in polynomial time.
\end{corollary}

\subsection{The Price of Variability}

To conclude our work, we take a step back from computational considerations and look at the implications of variability from a more general perspective.
Our goal is to quantify the conceptual loss of efficiency incurred by going from a fixed and known present bias to an unpredictable and variable one.
Similar to the price of uncertainty we define the {\em price of variability} as the following ratio.

\begin{definition}[Price of Variability]\label{def:pov}
	Given a task graph $G$ and a present bias set $B$, the price of variability is defined as $r(G,B^V)/\sup\{r(G,\{\beta\}) \mid \beta \in B\}$.
\end{definition}

It seems obvious that the price of variability depends closely on the structure of $G$ and $B$.
Nevertheless, we would like to find general bounds for the price of variability much like we did in Section~\ref{sec:upb} for the price of uncertainty.
As a first step, it is instructive to note that the price of uncertainty is a natural lower bound for the price of variability.
The reason for this is that each cost configuration that motivates an agent whose present bias varies arbitrarily in $B$ must also motivate an agent whose present bias is a fixed value from $B$.
Therefore it holds true that $r(G,B^V) \geq r(G,B)$, which immediately implies the stated bound. 
Sometimes this bound is tight.
Consider for instance Alice and Bob's scenario.
As we have shown in Section~\ref{sec:upb}, it is possible to construct a cost configuration $\tilde{c}$ verifying a price of uncertainty of $1$.
Using similar arguments, it is easy to see that $\tilde{c}$ remains motivating if we allow the present bias to vary, implying an identical price of variability.
However, for general instances of $G$ and $B$ this tight relation between the price of uncertainty and the price of variability is lost.
In fact, we can show that unlike the price of uncertainty, which has a constant upper bound of $2$, the price of variability may become arbitrarily large as the range of $B$ increases.

\begin{proposition}\label{prop:lpv}
	There exists a family of task graphs and present bias sets for which the price of variability converges to $\tau/2$.
\end{proposition}

\begin{proof}
	To obtain a price of variability close to $\tau/2$ we consider an occasionally unbiased agent with respect to the set $B = \{a,1\}$ for some $0 < a < 1/2$ such that $1/a$ is integral.
	Furthermore, we construct a task graph $G$ consisting of a directed path $v_0,v_1,\ldots,v_{1/a^2+1/a+2}$ whose edges are all of cost $1$.
	We call this the {\em main path}.
	In addition to the main path we introduce $1/a^2 + 1$ {\em shortcuts} via a common node $w$.
	Each shortcut $i$ with $0 \leq i \leq 1/a^2$ consists of two edges.
	The first edge goes from $v_i$ to $w$ for a cost of $2$ while the second edge goes from $w$ to $t$ for a cost of $1/a$.
	As source and target node, we choose $s = v_0$ and $t = v_{1/a^2+1/a+2}$.
	Figure~\ref{fig:lpv} shows a sketch of $G$.
	Note that some edges of the shortcuts are merged for the sake of a concise representation.
	
	\begin{figure}[t]
		\center
		\begin{tikzpicture}[scale=0.875, nst/.style={draw,circle,fill=black,minimum size=4pt,inner sep=0pt]}, est/.style={draw,>=latex,->}]
		
			\node[nst] (v2) at (2.5,0.5) [label=left:\small{$s$}] {};
			\node[nst] (v3) at (5,0) [label=below:\small{$v_1$}] {};
			\node[nst] (v4) at (7.5,0) [label=below:\small{$v_2$}] {};
			\node[nst] (v5) at (10,0) [label=below:\small{$v_{1/a^2+1/a+1}$}] {};
			\node[nst] (v6) at (12.5,0.5) [label=right:\small{$t$}] {};
			\node[nst] (v7) at (10,1) [label=above:\small{$w$}] {};
			
			\node at (8.75,0) {$\dots$};
			\node at (8.75,1) {$\dots$};			
			
			\path (v2) edge (2.5,0) 
			(2.5,0) edge[est] node [below] {\small{$1$}} (v3)
			(v3) edge[est] node [below] {\small{$1$}} (v4)
			(v4) edge (8.25,0)
			(9.25,0) edge[est] (v5)
			(v5) edge node [below] {\small{$1$}} (12.5,0)
			(12.5,0) edge[est] (v6)
			(2.5,1) edge node [above] {\small{$2$}} (5,1)
			(v2) edge (2.5,1)
			(v3) edge node [left] {\small{$2$}} (5,1)
			(v4) edge node [left] {\small{$2$}} (7.5,1)
			(v5) edge[est] node [left] {\small{$2$}} (v7)
			(5,1) edge (8.25,1)
			(9.25,1) edge[est] (v7)
			(v7) edge node [above] {\small{$1/a$}} (12.5,1)
			(12.5,1) edge[est] (v6);
			\end{tikzpicture}
		\caption{Task graph with a price of variability of $(1/a^2)/(2/a+2)$}\label{fig:lpv}
	\end{figure}
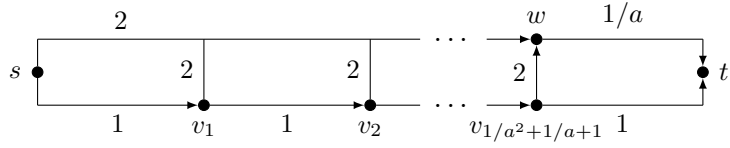
	
	The remainder of the proof has a similar structure to that of Proposition~\ref{prop:lpu}.
	We first argue that a reward of $2/a+2$ is sufficiently motivating for any agent with a fixed present bias of either $a$ or~$1$, implying $\sup\{r(G,\{a\}),r(G,\{1\})\} \leq 2/a+2$.
	We then show that no cost configuration $\tilde{c}$ can motivate an occasionally unbiased agent for a reward less than $1/a^2$, i.e., $r(G,\{a,1\}^V) \geq 1/a^2$.
	As a result, the price of variability must be at least $(1/a^2)/(2/a+2)$.
	Note that this term approaches $1/(2a) = \tau/2$ as $a \to 0$, which establishes the theorem.
	
	To see that a reward of $2/a+2$ is sufficiently motivating for an agent with a fixed present bias of $a$, assume such an agent is located at an arbitrary node $v_i$ with $0 \leq i \leq 1/a^2+1/a+1$.
	Her perceived cost for traversing $(v_i,v_{i+1})$ is at most $d_{a}(v_i,v_{i+1}) \leq 1 + a(2 + 1/a) = 2 + 2a$ if she plans to take the next shortcut at $v_{i+1}$.
	In the special case of $i = 1/a^2$, her perceived cost is $d_{a}(v_{1/a^2},v_{1/a^2 + 1}) = 1$ as she can reach $t$ directly via $(v_{1/a^2},v_{1/a^2 + 1})$.
	Either way, a reward of $r = 1/a(2 + 2a) = 2/a + 2$ covers her perceived cost for staying on the main path.
	In contrast, taking the immediate shortcut at $v_i$ has a perceived cost of $d_{a}(v_i,w) = 2 + a(1/a) = 3$.
	As we assume $a < 1/2$, the agent clearly perceives shortcuts as more expensive than the main path.
	Consequently, she follows the main path from $s$ to $t$ for a reward of $1/a(2 + 2a) = 2/a + 2$.
	
	Next consider an unbiased agent, i.e., an agent with a fixed present bias of~$1$.
	Such an agent strictly follows a cheapest path $P$ from $s$ to $t$.
	Furthermore, her perceived cost along $P$ never exceeds the total cost of $P$.
	Taking the first shortcut at $s$, we can bound the cost of a cheapest path from $s$ to $t$ from above by $2 + 1/a < 2/a + 2$.
	This implies that the unbiased agent successfully reaches $t$ for a reward of $2/a + 2$.
	
	It remains to show that no cost configuration $\tilde{c}$ can be motivating for all present bias configurations $\beta \in \{1,a\}^V$ if the reward is less than~$1/a^2$.
	For the sake of contradiction, assume such a $\tilde{c}$ exists.
	Note that the agent must not visit~$w$.
	The reason is that a reward of less than $1/a^2$ is not sufficient to make her traverse $(w,t)$ should her present bias become $\beta(w) = a$.
	However, to prevent the agent from taking a shortcut, $\tilde{c}$ must assign a cost greater than $1/a^2 - i$ to all shortcuts $i$ for $0\leq i \leq 1/a^2$.
	We proof this claim via an induction on $i$.
	For the base of the induction, let $i=1/a^2$.
	At $v_{1/a^2}$, exactly $2 + 1/a$ edges remain on the main path.
	Ignoring extra cost, there are two cheapest paths to~$t$, one along the main path and one along the current shortcut.
	Consequently, if we assume that the agent is momentarily unbiased, i.e., $\beta(v_{1/a^2}) = 1$, and therefore takes a cheapest path from $s$ to $t$, it becomes clear that $\tilde{c}$ must assign an extra cost greater than $0 = 1/a^2 - i$ to the current shortcut to prevent the agent from moving to~$w$.
	
	For the induction step let $i < 1/a^2$ and assume that each shortcut $j$ with $i < j \leq 1/a^2$ has an extra cost greater than $1/a^2 - j$.
	Our goal is to argue that the extra cost assigned to the shortcut $i$ must be greater than $1/a^2 - j$.
	When located at $v_{i}$, exactly $1/a^2 - i + 1/a + 2$ edges of the main path remain to~$t$.
	If the agent is currently unbiased, she perceives a cost of at least $1/a^2 - i + 1/a + 2$ for taking the main path.
	Clearly she cannot reduce this cost by planning to take a shortcut $j'$ with $j' > 1/a^2$.
	Should she consider a shortcut $j$ with $i < j \leq 1/a^2$ instead, she must first traverse $(j-i)$ edges of the main path.
	Together with the induction hypothesis her total perceived cost for such a plan is at least $(j-i) + 2 + 1/a + (1/a^2 - j) = 1/a^2 - i + 1/a + 2$.
	Consequently, her perceived cost for staying on the main path is at least $1/a^2 - i + 1/a + 2$.
	Since this exceeds the cost of shortcut $i$ by at least $1/a^2 - i$, we know that an extra cost greater than $1/a^2 - i$ must be assigned to the current shortcut to prevent the agent from walking onto $w$.
	This concludes the induction.
	
	We now know that all shortcuts $i$ with $0\leq i \leq 1/a^2$ have an extra cost greater than $1/a^2 - i$.
	By the same argument we have used in the inductive step it should be clear that each path from $s$ to $t$ has a cost of at least $1/a^2 + 1/a + 2$.
	Therefore if the agent is unbiased at $s$, we need a reward of $1/a^2 + 1/a + 2 > 1/a$ to motivate her.
	However, this contradicts our initial assumption on~$\tilde{c}$.
\end{proof}

Although Proposition~\ref{prop:lpv} implies that the price of variability can become substantially larger than the price of uncertainty, it should be noted that the task graph constructed in the proof of this proposition is close to a worst case scenario.
In particular, we can show that the price of variability cannot exceed $\tau + 1$, which is roughly twice the value obtained by Proposition~\ref{prop:lpv}.
To verify this upper bound, it is helpful to recall the proof of Theorem~\ref{thm:vaprox}.
In the process of establishing the approximation ratio of {\sc VariablePresentBiasApprox} we have argued that the cost configuration $\tilde{c}$ returned by the algorithm motivates any agent with a present bias configuration $\beta \in B^V$ for a reward of at most $(\tau + 1)r(G,\{\min B\})$.
Consequently, it holds true that $r(G,B^V )\leq (\tau + 1)r(G,\{\min B\})$, implying that the price of variability cannot exceed $\tau + 1$. 

\begin{corollary}
	The price of variability is at most $\tau + 1$.
\end{corollary}


\bibliographystyle{plainurl}
\bibliography{bib}

\end{document}